\documentclass[11pt]{article}

\usepackage{sn-preamble} 
\usepackage{tikz}
\usepackage{verbatim}
\usepackage{cancel}
\usepackage{enumitem}
\usepackage{bm}
\usepackage{caption}
\usepackage{changepage}   

\newcommand{\ignore}[1]{}

\newcommand{\dmax}{{d_{\max}}}
\newcommand{\Wmax}{{W_{\max}}}
\newcommand{\Rmax}{{R_{\max}}}

\newcommand{\bpi}{\boldsymbol{\pi}}

\newcommand{\findcandidatestem}{\textsc{FindCandidateStem}}
\newcommand{\genterm}{\textsc{GenerateCandidateStem}}
\newcommand{\genstem}{\textsc{generate-stem}}

\newcommand{\MQ}{\mathrm{MQ}}
\newcommand{\EQ}{\mathrm{EQ}}

\newcommand{\ALLPOS}{\textit{ALLPOS}}

\newcommand{\FCSM}{FCS_{\max}}

\newcommand{\FindRelevantVariable}{\textsc{FindRelevantVariable}}

\newcommand{\LearnDNF}{\textsc{LearnDNF}}

\newcommand{\Features}{\mathrm{Features}}

\title{
Faster exact learning of $k$-term DNFs\\
with membership and equivalence queries
}
\author{Josh Alman\thanks{Email: \texttt{josh@cs.columbia.edu}}
\\ \textsl{Columbia University} \and Shivam Nadimpalli \thanks{Email: \texttt{shivamn@mit.edu}} \\ \textsl{MIT} \and 
Shyamal Patel\thanks{Email: \texttt{shyamalpatelb@gmail.com}} \\ \textsl{Columbia University} \and Rocco A. Servedio\thanks{Email: \texttt{ras2105@columbia.edu}}\\ \textsl{Columbia University}}
\date{}

\begin{document}

\pagenumbering{gobble}
\maketitle

\begin{abstract}

In 1992 Blum and Rudich \cite{BlumRudich92} gave an algorithm that uses membership and equivalence queries to learn $k$-term DNF formulas over $\zo^n$ in time $\poly(n,2^k)$, improving on the naive $O(n^k)$ running time that can be achieved without membership queries \cite{Valiant:84}.
Since then, many alternative algorithms \cite{Bshouty:95,Kushilevitz:97,Bshouty:97,BBB+:00} have been given which also achieve runtime $\poly(n,2^k)$.

We give an algorithm that uses membership and equivalence queries to learn $k$-term DNF formulas in time $\poly(n) \cdot 2^{\tilde{O}(\sqrt{k})}$.
This is the first improvement for this problem since the original work of Blum and Rudich \cite{BlumRudich92}.

Our approach employs the Winnow2 algorithm for learning linear threshold functions over an enhanced feature space which is adaptively constructed using membership queries.
It combines a strengthened version of 
a technique that effectively reduces the length of DNF terms 
from the original work of  \cite{BlumRudich92} with a range of additional algorithmic tools (attribute-efficient learning algorithms for low-weight linear threshold functions and techniques for finding relevant variables from junta testing) and analytic ingredients (extremal polynomials and noise operators) that are novel in the context of query-based DNF learning.
\end{abstract}

\newpage

\pagenumbering{arabic}


\section{Introduction} \label{sec:intro}

\subsection{Background}

One of the most central and intensively studied problems in computational learning theory is that of learning an unknown \emph{Disjunctive Normal Form}, or \emph{DNF}, formula, i.e.~an OR of ANDs of Boolean literals.  The DNF learning problem was proposed in the seminal work of Valiant \cite{Valiant:84} that ushered in the modern field of computational learning theory; the problem subsequently garnered widespread interest due to the ubiquitous role of DNF as both a natural form of knowledge representation and a fundamental type of logical formula.  

Over the forty years that have passed since \cite{Valiant:84}, DNF learning has emerged as one of the most notorious problems in computational learning theory, as witnessed by the intensive research effort that has aimed at developing efficient DNF learning algorithms across a wide range of different learning models (see \Cref{sec:related-work} for a survey of some of this work). In particular, one important strand of research on learning DNF formulas has focused on the model of \emph{exact learning DNF with membership and equivalence queries}, which is the subject of the current paper.  
We give a detailed definition of this learning model in \Cref{sec:preliminaries}, but briefly, a \emph{membership query} is simply a black-box query for the value of $f$ at an input $x \in \zo^n$ that is specified by the learning algorithm.  In an \emph{equivalence query}, the learning algorithm specifies a hypothesis function $h: \zo^n \to \zo$ and the learning process ends (with the learner having succeeded) if $h \equiv f$, otherwise the learner is given an arbitrary \emph{counterexample} (a point $x \in \zo^n$ such that $h(x) \neq f(x)$).

While DNF learning has been intensively studied, the problem of learning an unknown and arbitrary $k$-term DNF formula turns out to be very challenging even for fairly small values of $k$, and progress on it has been quite limited.
In the original paper of Valiant \cite{Valiant:84},  de Morgan's law (which implies that any $k$-term DNF can be expressed as a $k$-CNF), feature expansion, and a simple reduction to learning conjunctions were used to give an $O(n^k)$-time algorithm that learns $k$-term DNF formulas in the exact learning model using equivalence queries only, without membership queries.  Improving on this result, Blum and Rudich \cite{BlumRudich92} (journal version in \cite{BlumRudich:95}) showed that by using both membership and equivalence queries, it is possible to achieve a $\poly(n,2^k)$ running time for learning $k$-term DNF formulas.  Note that with this running time, it is possible to learn $k=O(\log n)$-term DNF formulas in $\poly(n)$ time (as opposed to $k=O(1)$-term DNF formulas with the result of \cite{Valiant:84}).  Soon after \cite{BlumRudich92} quite a number of alternative algorithms achieving the same running time were given \cite{Bshouty:95,Kushilevitz:97,Bshouty:97,BBB+:00}, but none of those algorithms achieved a faster running time.\footnote{The above discussion focuses on the regime where $k$ is relatively small; for $k$ larger than roughly $n^{1/3}$, faster DNF learning algorithms are known, as described in \Cref{sec:related-work}.}

\subsection{Our result}

Our main result gives a more efficient learning algorithm for $k$-term DNF formulas, strictly improving on the results of \cite{BlumRudich92,Bshouty:95,Kushilevitz:97,Bshouty:97,BBB+:00}:

\begin{theorem} \label{thm:main}
There is an algorithm (\Cref{alg:Glorious-Learn-DNF}) that uses membership and equivalence queries to learn $k$-term DNF formulas over $\zo^n$ in time $\poly(n,2^{\tilde{O}(\sqrt{k})})$ in the exact learning model.
\end{theorem}

Achieving a faster running time than $\poly(n,2^k)$ (or equivalently, learning $k$-term DNF formulas in $\poly(n)$ time for $k=\omega(\log n)$) is a well-known open problem that has been mentioned in a number of previous works, see e.g.~\cite{BlumRudich:95,Kushilevitz:97,BBB+:00,BBBEncyclopedia}. \Cref{thm:main} gives the first progress on this problem since the original 1992 paper of Blum and Rudich \cite{BlumRudich:95}.  
Using \Cref{thm:main}, it is possible to learn $k=\tilde{O}((\log n)^2)$-term DNF formulas in $\poly(n)$ time.  

Our approach combines a wide range of different conceptual and algorithmic ingredients, including

\begin{itemize}

\item the attribute-efficient Winnow2 algorithm \cite{lit87} for learning low-weight, low-degree polynomial threshold functions; 
\item an adaptive feature expansion computed by an improved variant of an algorithmic technique from \cite{BlumRudich92} for effectively reducing the length of terms; 

\item the use of extremal (Chebyshev) polynomials to achieve a square-root degree savings in polynomial threshold function representations of DNF formulas; and 

\item the use of noise operators and algorithmic ideas from junta testing to identify small sets of ``useful'' variables.

\end{itemize}

While we will delve into more detail of how we use these techniques in \Cref{sec:techniques}, we briefly note that many of these steps are novel in the context of exact learning. In particular, while a number of prior works learn functions as a polynomial threshold function over some set of features, we are unaware of any earlier learning algorithms (in any model of learning) that \emph{adaptively} compute a feature expansion and learn a polynomial threshold function over the computed features. Also, to the best of our knowledge, ours is the first paper on membership query based learning that uses constructions based on extremal polynomials such as Chebyshev polynomials.  
We further remark that we are unaware of prior work that uses the Bonami-Beckner noise operator to facilitate finding ``useful'' variables for distribution-free or exact learning (although the  work of \cite{BLQT22}, on agnostically learning decision trees, does something similar in a uniform-distribution context, and related ideas have been also used in the context of (uniform-distribution) tolerant property testing of juntas \cite{DMN19,ITW21,nadimpalli2024optimal}).

We give a detailed overview of our approach in \Cref{sec:techniques}.

\subsection{Related work} \label{sec:related-work}

The prior works that are closest to ours, in terms of the learning problem (learning an arbitrary $k$-term DNF formula for $k$ as large as possible) and learning model (exact learning using membership and equivalence queries) being considered, are the aforementioned papers of \cite{BlumRudich92,Bshouty:95,Kushilevitz:97,Bshouty:97,BBB+:00}, all of which gave $\poly(n,2^{k})$-time algorithms.  These works used a broad range of different technical approaches including Bshouty's ``monotone theory'' \cite{Bshouty:95}, divide-and-conquer approaches \cite{Bshouty:97}, reductions to learning finite automata \cite{Kushilevitz:97}, and reductions to learning multiplicity automata \cite{BBB+:00}.   We remark that all of the approaches taken in \cite{BlumRudich92,Bshouty:95,Kushilevitz:97,Bshouty:97,BBB+:00} are combinatorial in nature; one of our innovations, as described in \Cref{sec:techniques}, is to combine combinatorial reasoning (in particular, an extension of the ideas in the original \cite{BlumRudich92} paper) with analytic ingredients, including arguments about extremal (Chebyshev) polynomials which previously proved useful in a different line of work on DNF learning as described below.

As was alluded to earlier, if $k$ is quite large then faster-than-$\poly(n,2^k)$-time algorithms have long been known, even without using membership queries.
Bshouty \cite{bsh96} used techniques for learning decision lists to exact-learn $k$-term DNF in time $2^{O(\sqrt{n \log k} (\log n)^{3/2})}$ with equivalence queries only, and (in the PAC model) Tarui and Tsukiji \cite{TaruiTsukiji:99} gave a boosting-based algorithm which achieves a similar runtime of $2^{O(\sqrt{n} \log n \log k)}$. Most recently,  Klivans and Servedio \cite{KlivansServedio:04jcss} combined a construction based on Chebyshev polynomials with Bshouty's decision list approach to give a $2^{O(n^{1/3} \log k \log n)}$-time exact learning algorithm using equivalence queries only.  We remark that even if $k=O(1)$, all of these approaches run in time at least $2^{n^{1/3} \log n}$, and hence for small $k$ they are much slower than the $\poly(n,2^k)$ runtime of the \cite{BlumRudich92,Bshouty:95,Kushilevitz:97,Bshouty:97,BBB+:00} approaches (and of our algorithm).

Returning to the membership and equivalence query model, we remark that a number of other prior works have considered both algorithms and hardness results for \emph{proper} exact learning (meaning that each hypothesis must itself be a DNF formula) of $k$-term DNF formulas with membership and equivalence queries for various different values of $k$; see, for example, \cite{Angluin92,Berggren93,BGHM96,PR96,HellersteinRaghavan:05}.

Finally, we mention that a wide range of other algorithms and hardness results have been given for many other variants of the DNF learning problem. A non-exhaustive list of the problem variants that have been considered would include learning under the uniform distribution, product distributions, and related distributions \cite{ver90,jac97,BJT:99,KlivansServedio:03ml,Gavinsky:03,Feldman07,Feldman12,de2014learning,ELR15,LLR19}; learning from uniform quantum superpositions of labeled examples \cite{BshoutyJackson:99,JTY:02}; and learning from labeled examples generated according to a  random walk over the Boolean hypercube $\{0,1\}^n$ \cite{BFH02,BMO+:05}.  
Many papers (including \cite{Valiant:84,PagalloHaussler89,Hancock:91,HancockMansour:91,AizensteinPitt91,AHP92,KushilevitzRoth:93,KMP:94,Khardon:94,AizensteinPitt:95,PR95,ABKKPR98,FeigelsonHellerstein98,Verbeurgt:98,dommispit99,SakaiMaruoka:00,BBB+:00,Servedio:04iandc,JacksonServedio:06long, JLSW11:dam, Sellie:08,Sellie-2009,HKSS12,ANPS25}) have also considered algorithms for learning a wide range of different restricted subclasses of DNF formulas, including  monotone DNF; read-once, read-twice, and read-$k$ DNF; disjoint DNF; satisfy-$j$ DNF; random DNF; and DNF in which all terms have the same length, among others.   These problems have been studied in a number of different models, both with and without membership queries.

\section{Our techniques} \label{sec:techniques}

At a very high level, our approach to learning an unknown $k$-term DNF $f$ is based on low-degree \emph{polynomial threshold functions}, referred to subsequently as PTFs. We view $f$ as a PTF of degree $\tilde{O}(\sqrt{k})$ which has integer coefficients that are not too large, of  total magnitude at most $W = 2^{\tilde{O}(\sqrt{k})}$; we refer to this total magnitude as the \emph{weight} of the PTF.  We use the well-known Winnow2 algorithm for learning low-weight linear threshold functions \cite{lit87} to learn based on this representation, via the simple observation that a low-weight low-degree polynomial threshold function can be viewed as a low-weight linear threshold function over an expanded feature set consisting of all low-degree monomials.  Recall that Winnow2 can learn a linear threshold function with total integer weight $W$ over $N$ features in time $\poly(N,W)$, using at most $\poly(W) \cdot \log N$ equivalence queries; see \Cref{thm:Winnow2} for a detailed statement.  (We remark that the $\log N $ dependence of Winnow2 in the number of equivalence queries it makes is often referred to as \emph{attribute-efficiency}; as we will see near the end of this section, Winnow2's attribute-efficiency plays a crucial role in the success of our overall approach.)  

However, there are several significant hurdles that need to be addressed in order to get such a strategy to work. 

\medskip
\noindent
{\bf First challenge: 
 High PTF degree.}
 The first and most obvious hurdle is that the PTF degree of a $k$-term DNF may be as large as $\Omega(k)$. To see this, we recall that for any $k \leq n^{1/3},$ the read-once ``tribes'' DNF which has $k$ ``tribes'' (terms) each of length $k^2$ is well known \cite{MinskyPapert:68} to have PTF degree $\Omega(k)$. How, then, can we view $f$ as a PTF of degree only $\tilde{O}(\sqrt{k})$?
 
To get around this barrier, a natural idea is to work with an expanded set of base features (going beyond the $n$ Boolean variables $x_1,\dots,x_n$), with the hope that $f$ can be expressed as a lower-degree PTF over the expanded feature space.
However, another obstacle to implementing such an idea quickly arises; this is the result of Razborov and Sherstov (Corollary~1.3 of \cite{RazborovSherstov10}), whose proof implies that for any $k \leq n^{1/3}$ and any fixed set of real-valued functions/features (not just low-degree monomials), if all $k$-term DNF formulas over $\zo^n$ can be expressed as linear threshold functions over those features, then the set must have at least $2^{\Omega(k)}$ elements.
This tells us that it will not be possible to ``obliviously'' give a feature expansion that will provide a learning algorithm with the desired $\poly(n,2^{\tilde{O}(\sqrt{k})})$ running time bound. Instead, a crucial aspect of our approach is to use membership queries to \emph{adaptively} construct a suitable set of features.

\medskip
\noindent
{\bf The solution:  Adaptively finding ``stems'' using membership queries.}
With the goal of achieving PTF degree $\tilde{O}(\sqrt{k})$, we take the following approach: For each term $T_i$, we say that a \emph{valid stem} for $T_i$ is a term $T'_i$ such that (a) $T'_i \subseteq T_i$ (i.e.~every literal in $T'_i$ also belongs to $T_i$), and (b) $|T_i \setminus T'_i| \leq 2k$, i.e.~$T'_i$ is only missing at most $2k$ literals from $T_i$. (Ideally we would like this parameter controlling the number of missing literals to be as small as possible; in \Cref{sec:findstem} we will see that our techniques allow us to take this parameter to be as small as $2k.$)
Note that if the length $|T_i|$ of a term $T_i$ is at most $2k$, then the empty term is a valid stem for $T_i$.
Our plan, at a very high level, is to find a valid stem $T'_i$ for each term $T_i$.  

To see why this would be helpful, consider first the special case in which every term in the DNF has length at most $2k$ (so the empty stem is all that we need; or equivalently, we do not need to use stems at all).  Using constructions of extremal polynomials which were first applied to DNF learning implicitly in the work of Tarui and Tsukiji \cite{TaruiTsukiji:99} and then explicitly by Klivans and Servedio \cite{KlivansServedio:04jcss}, it is not hard to show that any  $k$-term  DNF in which each term has length $2k$ can be expressed as a PTF of degree $O(\sqrt{k} \cdot \log k)$. In the general case, though, some or all terms of the DNF may have length greater than $2k$.  For each such term $T_i$, we attempt to discover a valid stem $T'_i$ of $T_i$; it is not too difficult to show (see \Cref{sec:structural})
that if we indeed had a valid stem for each term $T_i$, then a simple extension of the \cite{KlivansServedio:04jcss} PTF construction gives an ``augmented'' PTF of degree $\tilde{O}(\sqrt{k})$, where each ``augmented'' monomial in the augmented PTF consists of a suitable valid stem together with at most $\tilde{O}(\sqrt{k})$ of the original variables $x_1,\dots,x_n$.
Moreover, the weight of this augmented PTF is at most $2^{\tilde{O}(\sqrt{k})}$, as desired.

To find the desired valid stems, we build on an algorithmic idea from \cite{BlumRudich92}. Roughly speaking, they gave a membership query based procedure which, given as input a string $y$ that satisfies a \emph{unique} term $T_i$, finds a valid stem for $T_i$.  Of course, in order for this to be useful, one needs a way to find an assignment $y$ that satisfies a unique term; \cite{BlumRudich92} gave a randomized sub-procedure that succeeds in finding such an assignment with probability $2^{-O(k)}$, so by running this sub-procedure $2^{O(k)}$ times, it is possible to find a valid stem with high probability.  (During these repetitions, the \cite{BlumRudich92} approach also finds $n \cdot 2^{O(k)}$ other candidate stems, most of which are presumably not valid stems for any term of $f$.)  It is clear that  this $2^{O(k)}$ runtime is too slow for our purposes.  Instead, we use a slightly different algorithmic approach (see \Cref{alg:find-stem}).
Via a delicate analysis which carefully exploits a randomized ordering of variables used by our algorithm (see the proof of \Cref{lem:find-stem} in \Cref{sec:findstem}), we are able to show that our procedure has a non-trivial chance of significantly decreasing the number of terms satisfied by our initial input string $y$.  While the procedure may not ultimately reach an assignment satisfying a \emph{unique} term, intuitively we are able to show that with non-trivial probability, we reach a point at which all satisfied terms are ``sufficiently similar'' that it is possible to output a valid stem.  This improved stem-finding sub-procedure (this is essentially step~2 of \Cref{alg:find-stem}) succeeds in finding a valid stem with probability $2^{-\tilde{O}(\sqrt{k})},$ rather than $2^{-O(k)}$ as in \cite{BlumRudich92}.  So by repeating only roughly $2^{\tilde{O}(\sqrt{k})}$ times, our overall algorithm \Cref{alg:find-stem} is able to find a valid stem with extremely high probability.  (Like \cite{BlumRudich92}, our approach also finds many --- in our case, roughly $n \cdot 2^{\tilde{O}(\sqrt{k})}$ many --- other candidate stems, most of which are presumably not valid stems.)

So at this point, let us suppose that we have identified a set of $n \cdot 2^{\tilde{O}(\sqrt{k})}$ candidate stems, among which are valid stems for each of the $k$ terms of $f$.  Another challenge arises:

\medskip
\noindent
{\bf Second challenge:  Too many possible augmented monomials.}
Recall that the augmented PTFs mentioned earlier are composed of augmented monomials, each of which contains a single candidate stem together with at most $\tilde{O}(\sqrt{k})$ of the original variables $x_1,\dots,x_n$.  For a given candidate stem there are ${n \choose \tilde{O}(\sqrt{k})} \approx n^{\tilde{O}(\sqrt{k})}$ possible augmented monomials of this sort, so if we were to run the Winnow2 algorithm over a feature space consisting of all such possible augmented monomials, the running time of this ham-fisted approach would be at least $n^{\tilde{O}(\sqrt{k})}$ (since the running time of Winnow2 is at least linear in the dimension of its feature space).  While $n^{\tilde{O}(\sqrt{k})}$ is a faster running time than $\poly(n,2^k)$ for $k$ larger than roughly $(\log n)^2$, such a running time would be worse than the $\poly(n,2^k)$ runtime of \cite{BlumRudich92,Bshouty:95,Kushilevitz:97,Bshouty:97,BBB+:00} for smaller values of $k$; in particular, this would not give a $\poly(n)$ time algorithm for any values of $k$ that are $\omega_n(1)$, which would be a significant drawback.

So the second challenge we face, then, is that there are too many possible augmented monomials.  How can this be overcome?

\medskip
\noindent
{\bf The solution:  Using the noise operator and membership queries to identify a small set of auxiliary monomials for each stem.}
To deal with this challenge, let us first return to the same special case considered earlier, in which every term in the DNF has length at most $2k$.  In this case, since there are only $k$ terms and each term has length at most $2k$, there are at most $O(k^2)$ relevant variables in the entire DNF --- in other words, the DNF is an $O(k^2)$-junta.  Thus, it is natural to look to the junta testing literature for tools that may be useful; and indeed, in this situation it is possible to find all $\poly(k)$ many relevant variables by using membership queries to perform a simple binary search, as is frequently done in junta testing (see e.g. \cite{FKRSS03,Blaisstoc09,liu2018distribution, bshouty2019almost}).\footnote{We remark that in the context of $k$-junta testing, it is of paramount importance to have no dependence on $n$ at all in the query complexity. In our learning setting, it is fine to have even a $\poly(n)$ dependence on $n$ in the number of queries we make, so we can sidestep some of the difficulties that arise in junta testing with our membership query based search.}

Of course, in the real setting of an arbitrary $k$-term DNF, the function $f$ may not be an $O(k^2)$ junta (it may even depend on all $n$ variables); again, we must deal with long terms.  We do this by returning to the idea of stems: the idea of our approach is that for each candidate stem, we maintain a set $R_{T'}$ of auxiliary variables that are associated with that candidate stem, and 
\begin{quote}
($\star$) we only use, as features for Winnow2, augmented monomials which consist of some candidate stem $T'$ together with at most $\tilde{O}(\sqrt{k})$ variables \emph{all of which must belong to its associated set $R_{T'}$}. 
\end{quote}
(Upon first creating each candidate stem $T'$, the set $R_{T'}$ of auxiliary variables for it is initialized to the empty set.) 

The crucial properties that we must ensure are that no set $R_{T'}$ ever grows too large, and that eventually we get ``the right variables'' in the sets $R_{T'}$ for which $T'$ is an actual stem.  More precisely, our desiderata are that (a) every set $R_{T'}$ always has size at most $\poly(k)$, and (b) for each term $T_i$ of $f$, there is a candidate stem $T'_i$ that is a valid stem for $T_i$ for which the associated set $R_{T'_i}$ grows to eventually contain all the variables in $T_i \setminus T'_i$ (recall that if $T'_i$ is a valid stem of $T_i$, then $|T_i \setminus T'_i| \leq 2k$).  Given these desiderata, for each stem $T'$ the total number of augmented monomials of type ($\star$) described above is at most 
\[
\overbrace{n \cdot 2^{\tilde{O}(\sqrt{k})}}^{\small{\text{\# of candidate stems}}} \cdot \overbrace{{\poly(k) \choose \tilde{O}(\sqrt{k})}}^{\small{\text{\# of augmented monomials for a given candidate stem}}}=n \cdot 2^{\tilde{O}(\sqrt{k})};
\]
this is the number of features for Winnow2, so Winnow2 at least has a chance of running in $\poly(n,2^{\tilde{O}(\sqrt{k})})$ time as desired.

Item (a) above is easy to ensure simply by never growing a set $R_{T'}$ once it reaches a certain $\poly(k)$ size.
For item (b), in \Cref{sec:find-relevant-variables} we give an algorithm (\Cref{alg:find-relevant-vars}) which, very roughly speaking, given as input a pair $(T',R_{T'})$ where $T'$ is a stem of some term $T$ in $f$ and $R_{T'}$ is a set of auxiliary variables that does not yet contain all of the missing variables in $T \setminus T'$, updates $R_{T'}$ by adding to it a new variable from $T \setminus T'$.  (This is something of an oversimplification; the algorithm actually also requires a positive and a negative input for $f$ that both satisfy the stem $T'$ and also have some other property; see \Cref{lemma:findessentialvariables} for a precise statement.)  
While we defer a detailed description of the algorithm to \Cref{sec:find-relevant-variables}, we make the following remarks:  the algorithm crucially  

\begin{itemize}

\item Uses a \emph{restriction} of $f$, which we denote $f_{T'}$, that satisfies the stem $T'$ --- by working with such a restriction, we effectively convert $f$ into a function for which the term $T$ is now a ``short term'' (of length at most $2k$); and 

\item Applies \emph{Bonami-Beckner noise} (see \Cref{def:noise}) to the restricted function $f_{T'}$. Intuitively, the utility of Bonami-Beckner noise for us in this context is the following:  As suggested by the special case above, \Cref{alg:find-relevant-vars} finds the new variable using a binary search type procedure, which requires that we start with both a satisfying assignment (call this $y$) of some short term of $f_{T'}$ and an unsatisfying assignment (call this $z$) of $f_{T'}$. If all the terms of $f_{T'}$ are short, this will find us a new relevant variable appearing in some short term, as desired. Unfortunately, in the general case, $y$ and $z$ could be chosen adversarially in such a way that the binary search procedure would give us a variable from a long term. In turn, this is problematic as we could potentially pick up $\Omega(n)$ variables in this way. To circumvent this, we work with a ``noised'' version of $f_{T'}$ and perform binary search to find a variable which, when flipped, changes the value of the noised function significantly. We can then show that flipping any variable that is only involved in long terms can only change the noised function negligibly (cf. \Cref{claim:L14}).

\end{itemize}

\medskip
\noindent
{\bf Third challenge:  Putting the pieces together.}
At this point, we have an algorithm (our \Cref{alg:find-stem}) that given a positive example $y$ for the DNF outputs a candidate list of stems; and an algorithm for ``growing'' suitable sets of auxiliary variables for our candidate stems (our \Cref{alg:find-relevant-vars}). Moreover, we know that if we indeed find a valid stem for every term in $f$ and correctly identify the auxiliary variables for each valid stem, then $f$ can be written as an augmented PTF of degree $\wt{O}(\sqrt{k})$ and weight $2^{\wt{O}(\sqrt{k})}$. As such our final challenge is to piece these together with an appropriate algorithm for learning halfspaces in the exact learning model. Perhaps, the most natural approach to do this would be through the classic Perceptron algorithm.\footnote{We earlier alluded to the fact that we will use Winnow2, but we encourage the reader to humor us here to better understand why we made this choice.} Recall that the Perceptron algorithm makes at most $\poly(W,|\calF|)$ mistakes where $W$ is the weight of the halfspace and $|\calF|$ denotes the number of features. Promisingly, over a feature set arising from a single valid stem for every term of $f$ and a set of $\poly(k)$ auxiliary variables for each of these stems, Perceptron would then make at most $2^{\wt{O}(\sqrt{k})}$ mistakes, which is within our target limit. 

Now, the main difficulty we will face is that we are not given the relevant stems and sets of auxiliary variables up front; we will compute many invalid candidate stems throughout our algorithm. This is easiest to see in the simpler setting where we only wish to show that we can learn in time $n^{\wt{O}(\sqrt{k})}$. In this case, we can simply take the set of auxiliary variables for each stem to be $[n]$. Now whenever we run Perceptron and it fails to learn after making $\poly(W,|\calF|)$ many mistakes, it follows there was some example that satisfied a term in $f$ for which we have not yet computed a stem. However, because we do not know which example this is, we must run our procedure for finding stems on \emph{every} positive example that we saw in this run of perceptron. This, in turn, results in as many as $\poly(W,|\calF|) \cdot 2^{\wt{O}(\sqrt{k})} n$ features in the next iteration. Since we may have to repeat this $k$ times to find a stem for every term, we could end up with as many as 
    \[ \left( 2^{\wt{O}(\sqrt{k})} n \right)^k = 2^{\wt{O}(k^{3/2})} n^k\]
features, which is far too many for us to handle. 

\medskip
\noindent 
\textbf{The Solution: Winnow2 and Attribute Efficient Learning.} To circumvent these problems, we instead use the Winnow2 learning algorithm, which only makes at most $\poly(W, \log(|\calF|))$ many mistakes in our context. This logarithmic dependence on the number of features is precisely what allows us to circumvent the blow up in the number of features witnessed above. In particular, as in the discussion for Perceptron above, we will not run our stem-finding procedure (\Cref{alg:find-stem}) only once; rather, we run it \emph{once for each positive example that has been received so far in the most recent run of Winnow2.} This again creates many invalid candidate stems and the number of candidate stems that are generated in the $t$-th execution of 2(ii) depends on the  number of positive examples that were received so far, and hence on the number of equivalence queries (i.e.~mistakes) that were made by the Winnow2 algorithm so far. The attribute-efficiency of Winnow2 turns out to be essential in giving a bound on this number of equivalence queries (mistakes) made which is strong enough for the overall approach to work. 

To give more details, our final algorithm (\Cref{alg:Glorious-Learn-DNF}) roughly works by alternating between:

\begin{itemize}

\item [(1)] running Winnow2 over the current set of features (corresponding to augmented monomials over the current candidate stems and their associated sets of auxiliary variables), and

\item [(2)] expanding the current feature set and restarting Winnow2 with the new expanded feature set.  Expanding the feature set is done in two different ways.  The first (i) (corresponding to step~3(b) of \Cref{alg:Glorious-Learn-DNF}) is by growing the set $R_{T'}$ of auxiliary variables for some current candidate stem $T'$ using \Cref{alg:find-relevant-vars} (we attempt to do this each time Winnow2 receives a negative example). The second (ii) (corresponding to step~3(c) of \Cref{alg:Glorious-Learn-DNF}) is by adding new candidate stems via \Cref{alg:find-stem} (we do this each time the current run of Winnow2 has made ``too many equivalence queries'' using its current feature set).

\end{itemize}

Throughout \Cref{alg:Glorious-Learn-DNF}, we enforce a mistake bound of $\poly(n) 2^{\wt{O}(\sqrt{k})}$ on Winnow2, i.e.~we halt the run of Winnow2 if it makes more than this many mistakes (equivalence queries). This will turn out to always be a large enough bound such that Winnow2 would successfully learn if the set of features was adequate to express the target DNF as an augmented PTF. It then remains to show that we do not need to restart Winnow2 on a new feature set too many times.  The idea for why this holds is roughly as follows. Our analysis shows that each time we do 2(ii), with extremely high probability we find a valid stem for a new term of $f$ for which we did not previously have a valid stem, so 2(ii) cannot be carried out more than $k$ times.  

Given this, \Cref{alg:Glorious-Learn-DNF} cannot construct too many candidate stems, and given a bound on the total number of candidate stems, it can be shown that 2(i)  (growing a set $R_{T'}$) also cannot happen too often, since there are not too many candidate stems $T'$ and each $R_{T'}$ never grows too large. This lets us upper bound the total number of times that either (2)(i) or 2(ii) can be carried out; in turn, this means that (1) (a new execution of Winnow) does not need to be started too many times, as desired.


\section{Preliminaries} \label{sec:preliminaries}

\noindent {\bf Exact learning with queries.}
We first recall the model of \emph{exact learning from equivalance queries only}, which was introduced and studied in the foundational work of Angluin \cite{Angluin:87}.  (As is discussed in \cite{lit87} and is well known, this model is completely equivalent to the \emph{on-line mistake bound (OLMB)} learning model that was introduced by Littlestone in another seminal paper \cite{lit87}; we will use this connection soon.)  In this model, when learning a concept class ${\cal C}$ (i.e.~a class of Boolean functions mapping $\{0,1\}^n$ to $\{0,1\}$), learning proceeds in a sequence of \emph{rounds}.  At the start of each round, the learning algorithm has its \emph{current hypothesis} which is some function $h: \{0,1\}^n \to \{0,1\}$ ($h$ need not belong to ${\cal C}$). 
In each round, the learning algorithm makes an \emph{equivalence query}, which amounts to asking an \emph{equivalence oracle} the question ``is $h$ logically equivalent to the target function $f$ (in the sense that $h(x)=f(x)$ for all $x \in \{0,1\}^n$)?''  If the answer is yes then the learner has succeeded and the learning process ends; otherwise the equivalence oracle returns an arbitrary \emph{counterexample}, i.e.~a string $x \in \zo^n$ such that $h(x) \neq f(x)$, and the learning algorithm can update its current hypothesis before the next round begins.

It is known that learning from equivalence queries only is at least as hard as learning in the distribution-free PAC learning model \cite{Valiant:84}.  Most research in the exact learning model focuses on an augmented variant of the model, called \emph{exact learning from membership and equivalence queries}, which was introduced and studied in \cite{Angluin:87,Angluin:88}.  In this model, at each round the learning algorithm may either ask an equivalence query, which works as described above, or else it may opt to ask a \emph{membership query} (black-box query) on some point $x \in \{0,1\}^n$ of its choosing. If the learning algorithm makes a membership query on $x$, it receives the value value $f(x)$, and can update its current hypothesis in light of this new information.  The complexity of a learning algorithm in this framework is measured by the total number of rounds (membership queries plus equivalence queries) that it takes to exactly learn the target function $f$, as well as the total running time of the algorithm across all rounds.

\medskip

\noindent {\bf Attribute-efficient learning.}
A learning algorithm for a concept class ${\cal C}$ of Boolean functions over $\zo^n$ is said to be \emph{attribute-efficient} if, roughly speaking, when it it learns a function in ${\cal C}$ that depends only on $k \ll n$ relevant variables, the number of queries that it makes scales as $o(n)$ (see \cite{lit87,BHL:95,bshhel98,Kivinen2008}).
One of the earliest and best-known attribute-efficient learning algorithms, for the class of low-weight linear threshold functions, is Littlestone's Winnow2 algorithm \cite{lit87} which learns from equivalence queries only.  In particular, as an easy consequence of Theorem~9 of \cite{lit87}, we have the following well-known result:

\begin{theorem} [Winnow2 performance guarantee] \label{thm:Winnow2}
The Winnow2 algorithm is an algorithm for learning from equivalence queries which has the following performance guarantee: For any target function $f$ which is a linear threshold function $f(x)=\mathbf{1}[w \cdot x  \geq \theta]$ over $\zo^n$, where $w \in \Z^m$ has $\sum_{i=1}^m |w_i| \leq W$, the Winnow2 algorithm successfully performs exact learning of the target function in at most $\poly(W) \cdot \log n$ rounds, and its running time per round is $O(n)$.
\end{theorem}

(Using the aforementioned well-known equivalence between exact learning from equivalence queries only and online mistake-bound learning, we will often refer to the number of equivalence queries that Winnow2 uses as the number of ``mistakes'' that it makes.)

\medskip
\noindent
{\bf Notation and terminology regarding DNFs.}
We will frequently view a \emph{term} (conjunction) as being a set of literals from $\{x_1,\dots,x_n,\overline{x_1},\dots,\overline{x_n}\}$, and a DNF formula $f$ as being a set of terms. For example, we write ``$T \in f$'' to indicate that $T$ is a term in the DNF $f$, and we write ``$T_1 \subseteq T_2$'' to indicate that every literal in $T_1$ appears in $T_2$. Throughout this paper the Boolean variables $x_1,\dots,x_n$ take values in $\{0,1\}$, and likewise terms $T$ and DNF formulas $f$ take values in $\{0,1\}$ (where 1 corresponds to True).

We will use the following parameter setting throughout the paper:
\[
\tau:=1000k.
\]
Intuitively, $\tau$ is the threshold length for a term to be considered ``short'' in our analysis.

Given a term $T$, the \emph{length} of $T$, denoted $|T|$, is simply the size of the set of literals it contains (in the notation of the previous paragraph, it is $|T|$).  Given a DNF formula $g=T_1 \vee \cdots \vee T_s$ and an integer $L$, as described earlier we write $g_{\leq L}$ to denote the sub-DNF of $g$ that contains only the terms of width at most $L$, and similarly we write $g_{>L}$ to denote the sub-DNF of $g$ that contains only the terms of width greater than $L$, so we have 
\[
g ~=~ g_{\leq L} ~ \vee ~ g_{>L}.
\]

Additionally, for a $y \in \zo^n$, we will let $\calT_g(y)$ denote the set of terms in $g$ that are satisfied by $y$.

\section{Finding valid stems} \label{sec:findstem}
We begin by defining the ``stem of a term.''

\begin{definition}[Stem of a Term]
	Given terms $T', T$, we say that $T'$ is a \emph{valid stem} of $T$ if $T' \subseteq T$ and $|T \setminus T'| \leq 2k$. It will also be convenient to call an arbitrary term $T'$ a \emph{candidate stem}, but this has no semantic meaning beyond $T'$ being a term.
\end{definition}

Looking ahead, valid stems will be useful because they effectively allow us to reduce the size of long terms to be at most $O(k)$. In particular, we will show in \Cref{sec:structural} that using valid stems we can express any $k$-term DNF as a (feature expanded) PTF with degree $\wt{O}(\sqrt{k})$ and integer weights of total magnitude $2^{\tilde{O}(\sqrt{k})}$. This  matches what we would expect for a DNF whose terms all have length at most $O(k)$.

We now turn to show that we can find valid stems of long terms. At a high level, our approach is builds on Blum and Rudich's algorithm and analysis in \cite{BlumRudich92}. We start with the following simple procedure:

\bigskip
\begin{algorithm}[H]
\addtolength\linewidth{-2em}

\vspace{0.5em}

\textbf{Input:} Query access to $f: \zo^n \rightarrow \zo$, $y$ such that $f(y) = 1$\\[0.25em]
\textbf{Output:} A pair consisting of a term and an empty set

\

\genterm($f, y$):
\begin{enumerate}
	\item Let $I \subseteq [n]$ denote the indices $i$ such that $f(y^{\oplus i}) = 0$ 
	\item Let $T'$ denote the term with literals $\{x_i: i \in I, y_i = 1\} \cup \{\overline{x}_i: i \in I, y_i = 0\}$.
        \item Return $(T',\emptyset)$
\end{enumerate}

\caption{A simple procedure to find a subset of the literals appearing in a term}
\label{alg:find-stem}
\end{algorithm}

\bigskip

Throughout this section, we will ignore the $\emptyset$ appended to the term in the output of \genterm{} and simply pretend that the procedure outputs $T'$. (The $\emptyset$ is the initialization of the set $R_{T'}$ of auxiliary variables that we will need later in the proof, as mentioned in \Cref{sec:techniques}.) 

The \genterm{} procedure is also used by Blum and Rudich in their analysis. In particular, they crucially note that if $y$ satisfies a unique term $T$ of $f$, then \genterm$(f,y)$ outputs a valid stem for $T$. Indeed, a simple rephrasing of their proof of this follows by considering the following definition that will be useful for us later:

\begin{definition}[Protected Set]
\label{def:protected}
	Given a string $y \in \zo^n$ and a DNF $f$, we say that the \emph{protected set for $y$}, denoted $P(y)$, is formed as follows: For each term $T \in f$ that $y$ does not satisfy let $\ell_i$ be the literal in $T$ with smallest index $i \in [n]$ that $y$ does not satisfy. The protected set $S \subseteq [n]$ is the set of all these indices.
\end{definition}

Note that for any index $j \not \in P(y)$ with a corresponding literal appearing in $T$ (the unique term $y$ satisfies), we must have that $f(y^{\oplus j}) = 0$; moreover, every index in $I$ is easily seen to correspond to a literal that belongs to $T$. So the term $T'$ that $\genterm(f,y)$ outputs will consist only of literals of $T$, and it will be missing $|P(y)| \leq s-1$ literals of $T$.

Blum and Rudich's algorithm then goes on to give a randomized procedure to find a point that uniquely satisfies $T$ with probability $2^{-O(k)}$; by repeating this procedure $2^{O(k)}$ times, they are thus able to find a valid stem (as well as many other candidate stems).  Of course, this will be too slow for our purposes. Instead, we will show that outputting valid stems can be done much more quickly by the following \findcandidatestem~algorithm (\Cref{alg:find-stem}):

\bigskip

\begin{algorithm}[H]
\addtolength\linewidth{-2em}

\vspace{0.5em}

\textbf{Input:} Query access to $f: \zo^n \rightarrow \zo$, $y$ such that $f(y) = 1$\\[0.25em]
\textbf{Output:} A list $\calL_y$ of pairs each consisting of a candidate stems and an empty set

\findcandidatestem($f, y$):
\begin{enumerate}
	\item $\calL_y \gets \emptyset$
	\item Repeat $2^{\wt{O}(\sqrt{k})} \log(n)$ times:
		\begin{enumerate}
			\item $\bz_0 \gets y$
			\item Draw a random permutation $\bpi \in S_n$ 
			\item For $i = 0, 1, \dots, n$:
			\begin{enumerate}
				\item  
                Let $(T',\emptyset) = \genterm(f,\bz_i).$
                If $T'(y)=1$ then add $(T',\emptyset)$ to $\calL_y$.
				\item If $f(\bz_i^{\oplus \bpi(i)}) = 1$, then set $\bz_{i+1} \gets \bz_i^{\oplus \bpi(i)}$, otherwise set $\bz_{i+1} \leftarrow \bz_i$. (We refer to setting $\bz_{i+1} \leftarrow \bz_i^{\oplus \bpi(i)}$ as \emph{flipping} index $\bpi(i)$.)
		\end{enumerate}
				
		\end{enumerate}

	\item Return $\calL_y$
\end{enumerate}

\caption{A procedure for finding stems}
\label{alg:find-stem}
\end{algorithm}

\bigskip

We remark that this algorithm also bears a resemblance to Blum and Rudich's ``Random-sweep'' procedure. There are some significant differences, though: unlike Random-sweep we consider flipping the bits in $y$ one at a time in a \emph{random order} (this is very important for our analysis, as will become clear in \Cref{sec:main-lemma}), and we always flip an index $j$ so long as the resulting assignment continues to satisfy $f$. 

The main guarantee that we will need from \findcandidatestem{}, which is the main result of this section, is the following:

\begin{lemma}
\label{lem:find-stem}
    Let $f$ be a $k$-term DNF and let $y \in \zo^n$ satisfy $f_{> \tau}$ and not $f_{\leq \tau}$. 
    The algorithm $\findcandidatestem(f,y)$ runs in time $\poly(n) \cdot 2^{\tilde{O}(\sqrt{k})}$ and outputs a set of at most 
    \begin{equation}
    \FCSM:= n \log(n) \cdot 2^{\tilde{O}(\sqrt{k})}
    \end{equation}
candidate stems satisfied by $y$ such that with probability $1 - n^{-\Omega(k)},$ at least one of them is a valid stem for a term $T \in f$ that $y$ satisfies.
\end{lemma}

The rest of this section is dedicated to proving \Cref{lem:find-stem}.

\subsection{Setup}

Our chief goal will be to show that line $(2)$ of \Cref{alg:find-stem} outputs a valid stem with probability $2^{-\wt{O}(\sqrt{k})}$.
To begin reasoning towards this end, we note that so long as we do not flip a variable in the protected set $P(y)$, we will slowly be decreasing the number of terms $\bz_i$ satisfies in the loop on line $2(c)$. Intuitively, it will be useful for us to satisfy fewer terms, so we will be interested in the event that we do not flip a variable in the protected set. (Indeed, if we satisfied a unique term of $f$, we would be done.)

\begin{lemma}
\label{lem:sat-term-decreasing}
    Let $f$ be a DNF of size $k$ and $y \in \zo^n$ be such that $y$ satisfies $f_{> \tau}$ and not $f_{\leq \tau}$. Suppose that $i$ is such that $(\bz_i)_a = y_a$ for all $a \in P(y)$, then
	\[\calT_f(\bz_{i}) \subseteq \calT_f(\bz_{i-1}).\]
	(Recall that $\calT_f(x)$ is the set of terms in $f$ that are satisfied by assignment $x$.)
\end{lemma}

\begin{proof}
Since no coordinate in $P(y)$ was flipped before time $i$, neither the set $\calT_f(\bz_i)$ nor the set $\calT_f(\bz_{i-1})$ can contain any term $T$ that $y$ does not satisfy. So consider a term $T$ that $y$ satisfies.  If $T \not \in \calT_f(\bz_{i-1})$, then at some time step $j \leq i-1$ an index $a \in [n]$ must have been flipped that resulted in $T$ being unsatisfied. Since this coordinate is never flipped again, it follows that $T \not \in \calT_f(\bz_i)$, as desired.
\end{proof}

As a corollary, we note that if a coordinate $j$ did get flipped by the time we reach $\bz_i$, then assuming that no coordinates in the protected set for $y$ have been flipped, $j$ must be missing from all terms that are satisfied by $\bz_i$:

\begin{corollary} 
\label{cor:flipped-coords-missing}
    Let $f$ be a DNF of size $k$ and $y \in \zo^n$ be such that $y$ satisfies $f_{> \tau}$ and not $f_{\leq \tau}$. If $(\bz_i)_a = y_a$ for all $a \in P(y)$, $T$ is a term in $\calT_f(\bz_i)$, and $j \in [n]$ is such that $(\bz_i)_j \not = y_j$, then neither $x_j$ nor $\overline{x_j}$ appear in $T$.
\end{corollary}

\begin{proof}
Let $t \leq i$ denote the first time at which $(\bz_t)_j \not = y_j$. Applying \Cref{lem:sat-term-decreasing}, we note that $\bz_{t}$ and $\bz_{t-1}$ both satisfy $T$. As $(\bz_{t-1})_j \not = (\bz_t)_j$, it must follow that $T$ has no literal corresponding to $j$.
\end{proof}

We now consider the following two definitions:
 
\begin{definition}[Stripped Terms]
	Given a term $T$ and a set $S \subseteq [n]$, we say the \emph{stripped term}, denoted $T \setminus S$, is defined by taking $T$ and removing all literals $x_i$ and $\overline{x_i}$ where $i \in S$.  Given a set of terms $\calT$ and a set of indices $S$, we will write $\calT \setminus S$ to denote the set of stripped terms $\{T \setminus S: T \in \calT\}$
\end{definition}

\begin{definition}[Unanimous Indices]
	Given a set of terms $\calT$, we say the set of \emph{unanimous indices $U := U(\calT)$ of $\calT$} is the set of all indices $i \in [n]$ such that either every term in $\calT$ contains the literal $x_i$ or every term in $\calT$ contains $\overline{x_i}$.
\end{definition}

We will frequently consider the set 
\begin{equation} \label{eq:keyset}
\calT_f(\bz_i) \setminus (P(y) \cup \bU_i), \text{~where $\bU_i$ is the set of unanimous indices of $\calT_f(\bz_i)$}.
\end{equation}
In words, this is the set of all terms satisfied by $\bz_i$ after both the elements of the protected set and the unanimous indices of those terms have been stripped away. These should be thought of as the coordinates that we can ``safely flip'' in line 2(c)ii after the $i$th step of the loop without flipping an element of the protected set $P(y)$.

Crucially, any index with a literal that occurs in a term of $\calT_f(\bz_i) \setminus (P(y) \cup \bU_i)$ has not yet appeared under $\bpi$ by the $i$th iteration of the loop:

\begin{lemma}
\label{lem:stripped-indices-flippable}
  Let $f$ be a DNF of size $k$ and $y \in \zo^n$ be such that $y$ satisfies $f_{> \tau}$ and not $f_{\leq \tau}$. Fix an index $i$ such that $(\bz_i)_a = y_a$ for all $a \in P(y)$. Set $j \in [n]$ to be such that there is a corresponding literal in some stripped term in $\calT_f(\bz_i) \setminus (P(y) \cup \bU_i)$.
Then $\bpi(t) \not = j$ for all $t < i$.
\end{lemma}

\begin{proof}
    Towards a contradiction, suppose that $\bpi(t) = j$ for some $t < i$. Note that we cannot have flipped $j$ on line 2(c)ii, as if we had, then \Cref{cor:flipped-coords-missing} tells us that no term in $\calT_f(\bz_i)$ would contain a literal involving $j$. Thus, we must not have flipped $j$, which means that $f(\bz_t^{\oplus j}) = 0$. Note now that $j \not \in \bU_t$, since $j \not \in \bU_i$ and $\bU_i \supseteq \bU_t$ by \Cref{lem:sat-term-decreasing}. If some term in $\calT_f(\bz_t)$ did not contain a literal involving $j$, then we would have $f(\bz_t^{\oplus j}) = 1$, which is a contradiction; so it must be the case that there exists both a term in $\calT_f(\bz_t)$ containing $x_j$, and a term in $\calT_f(\bz_t)$ containing $\overline{x_j}$. However, this cannot happen either as $\bz_t$ cannot satisfy both of these terms. Since we have reached a contradiction in all cases, this completes the proof.
\end{proof}

We also quickly record the fact that under our usual assumption that no coordinates in the protected set for $y$ have been flipped so far, the set of indices of literals in a stripped term in $\calT_f(\bz_i) \setminus (P(y) \cup \bU_i)$ is non-increasing.

\begin{corollary}
\label{cor:stripped-terms-decrease}
    Let $f$ be a DNF of size $k$ and $y \in \zo^n$ be such that $y$ satisfies $f_{> \tau}$ and not $f_{\leq \tau}$. Suppose that $i$ is such that $(\bz_i)_a = y_a$ for all $a \in P(y)$, then
        \[\calT_f(\bz_i) \setminus (P(y) \cup \bU_i) \subseteq \calT_f(\bz_{i-1}) \setminus (P(y) \cup \bU_{i-1}) \]
    where $\bU_i$ denotes the set of unanimous indices among $\calT_f(\bz_i)$.
\end{corollary}

\begin{proof}
    Note that by \Cref{lem:sat-term-decreasing}, it follows that $\calT_f(\bz_i) \subseteq \calT_f(\bz_{i-1})$. As mentioned in the previous proof, this implies that $\bU_{i-1} \subseteq \bU_i$. Combining these two inclusions yields the corollary.
\end{proof}

We next observe that under our usual assumption that no coordinates in the protected set for $y$ have been flipped so far, flipping any variable $j \in \bU_i \setminus P(y)$ will result in $f(\bz_i^{\oplus j}) = 0$.

\begin{lemma}
\label{lem:uni-vars-flip-zero}
Let $f$ be a DNF of size $k$ and $y \in \zo^n$ satisfy $f_{> \tau}$ and not $f_{\leq \tau}$. Let $i$ be such that $(\bz_i)_a = y_a$ for all $a \in P(y)$. 
Then we have that $f(\bz_i^{\oplus j})=0$ for all $j \in \bU_i \setminus P(y)$.
\end{lemma}

\begin{proof}
    Consider a term $T \in f$ that is not satisfied by $\bz_i$. If $T$ is not satisfied by $\bz_0$ then there is a variable in $P(y)$ certifying this. Thus, $\bz_i^{\oplus j}$ cannot satisfy any term that $\bz_0$ doesn't satisfy. On the other hand, if $T$ is satisfied by $\bz_0$, then there must be some $t < i$ such that $\bpi(t)$ certifies that $\bz_i$ doesn't satisfy $T$. That said, by \Cref{cor:flipped-coords-missing}, we have that no variable in $\bU_i$ was flipped at time $t < i$. Thus, $\bpi(t) \not \in \bU_i$
    and $\bz_i^{\oplus j}$ again does not satisfy $T$. Thus, in either case, we have that flipping any variable $j \in \bU_i \setminus P(y)$ will not cause $\bz_i$ to satisfy a $T \not \in \calT_f(\bz_i)$. Since every term in $\calT_f(\bz_i)$ has a literal corresponding to the value of $(\bz_i)_j$, it then follows that $f(\bz_i^{\oplus j}) = 0$ for every $j \in \bU_i \setminus P(y)$ as claimed.
\end{proof}

We will need one last ingredient. This is a lemma which promises that under our usual assumption that no coordinates in the protected set for $y$ have been flipped so far, \genterm{} will output a valid stem when run on $\bz_i$ if there exists a stripped term in $\calT_f(\bz_i) \setminus (P(y) \cup \bU_i)$ of length at most $k$.

\begin{lemma}
\label{lem:genterm-output-stem}
Let $f$ be a DNF of size $k$ and $y \in \zo^n$ satisfy $f_{> \tau}$ and not $f_{\leq \tau}$. Let $i$ be such that $(\bz_i)_a = y_a$ for all $a \in P(y)$ and consider the $i$th step of the loop on line 2.(c).  
If $\calT_f(\bz_i) \setminus (P(y) \cup \bU_i)$ contains a stripped term of length at most $k$, then $\genterm(f,\bz_i)$ outputs a valid stem of a term satisfied by $y$.
\end{lemma}

\begin{proof}
    Let $T_{\text{stripped}} \in \calT_f(\bz_i) \setminus (P(y) \cup U_i)$ be the stripped term of length at most $k$. Note that by \Cref{lem:uni-vars-flip-zero}, flipping an index $j \in \bU_i \setminus P(y)$ will result in $f(\bz_i^{\oplus j}) = 0$.
    Now let $T_{\text{unstripped}}$ denote a term of $T$ corresponding to $T_{\text{stripped}}$. As $f(\bz_i^{\oplus j}) = 0$ for every $j \in \bU_i \setminus P(y)$, we conclude that \genterm$(f,\bz_i)$ outputs a term that includes all literals in $T_{\text{unstripped}}$ involving indices in $\bU_i \setminus P(y)$. By construction there are then at most $|T_{\text{stripped}}| + |P(y)| \leq 2k$ literals missing from the term output by $\genterm(f,\bz_i)$ in $T_{\text{unstripped}}$. On the other hand, $\genterm(f,\bz_i) \subseteq T_{\text{unstripped}}$, as line $2$ of $\genterm(f,\bz_i)$ will not add a literal corresponding to an index $j \in [n]$ with no literal in $T_{\text{unstripped}}$ since $\bz_i^{\oplus j}$ satisfies $T_{\text{unstripped}}$. Finally, by \Cref{lem:sat-term-decreasing}, $y$ satisfies $T_{\text{unstripped}}$. Thus, we have that \genterm$(f,\bz_i)$ outputs a valid stem as promised.
\end{proof}

\subsection{The main lemma} \label{sec:main-lemma}
We can now put the above results together to prove our main lemma underlying \Cref{lem:find-stem}, which states that a particular iteration of the loop on Line $2$ of \findcandidatestem{} will indeed output a valid stem with non-negligible probability:

\begin{lemma}
\label{lemma:find-stem-single-round}
    Suppose that $f$ is a size $k$ DNF and $y \in \zo^n$ satisfies $f_{> \tau}$ and not $f_{\leq \tau}$. Then any particular iteration of the while loop on Line $2$ of $\findcandidatestem(f,y)$ outputs a valid stem for some term $T$ satisfying $y$ with probability $2^{-\wt{O}(\sqrt{k})}$.
\end{lemma}

\begin{proof}
    We will prove the following claim by induction:
    \begin{claim}
    \label{claim:inductive-find-stem}
        Let $f$ be a size $k$ DNF and $y \in \{0,1\}^n$ satisfy $f_{> \tau}$ and not $f_{\leq \tau}$. Consider some iteration of the while loop on line $2$ of $\findcandidatestem(f,y)$ and fix an $i \in \{0,1, \dots n\}$. Moreover, suppose $(i)$ $\genstem(f,\bz_i)$ fails to output a valid stem in round $i$ and $(ii)$ $(\bz_i)_a = y_a$ for all $a \in P(y)$. Let $\bU_i$ denote the set of unanimous indices of $\calT_f(\bz_i)$. Assuming that $\genstem(f,\bz_i)$ did not output a valid stem, it then follows that one of the future iterations outputs a valid stem with probability at least $e^{-100\sqrt{|\calT_f(\bz_i) \setminus (P(y) \cup \bU_i)|}\log(k)}$.
    \end{claim}

\Cref{claim:inductive-find-stem} easily yields \Cref{lemma:find-stem-single-round} as follows: Since $\bz_0=y$, $\genstem(f,\bz_0)$ either does or does not output a valid stem. If it does, then we are done. If it does not, we can apply \Cref{claim:inductive-find-stem} with $i=0$ to conclude that some subsequent execution of $\genstem(f,\bz_{i'})$ outputs a valid stem with probability at least 
        \[e^{-100 \sqrt{|\calT_f(\bz_0) \setminus (P(y) \cup \bU_0)|}\log(k)} \geq e^{-100\sqrt{k}\log(k)}\]
    as desired. 

We now turn to prove the claim.
    \begin{proof}[Proof of \Cref{claim:inductive-find-stem}]
As promised, we will prove the statement by induction on $|\calT_f(\bz_i) \setminus (P(y) \cup \bU_i)|$. For the base case, suppose that $|\calT_f(\bz_i) \setminus (P(y) \cup \bU_i)| = 1$. It then follows that \genstem$(f,\bz_i)$ finds a valid stem, which means that \Cref{claim:inductive-find-stem} holds vacuously since its condition $(i)$ is violated.
To see that \genstem$(f,\bz_i)$ finds a valid stem,
fix a term $T$ satisfied by $\bz_i$. We claim that $T \setminus (P(y) \cup \bU_i) = \emptyset$. Indeed, for any literal in $T$ outside of $P(y)$ and $\bU_i$, there must exist a term of $\calT_f(\bz_i)$ that does not include that literal, as otherwise the corresponding index would appear in $\bU_i$. But this violates $|\calT_f(\bz_i) \setminus (P(y) \cup \bU_i)| = 1$. Thus, the stripped terms all have length $0$ and we output a stem by \Cref{lem:genterm-output-stem}.

We now move on to the inductive step; so we suppose that the claim is true for all $|\calT_f(\bz_i) \setminus (P(y) \cup \bU_i)| \leq m - 1$, and we wish to prove the claim when this set has size $m$. Our proof uses the following objects:

\begin{itemize}
\item Let $R \subseteq [n]$ denote the set of all indices occurring in either $P(y)$ or in some literal in 
\[\bigcup_{T \in  \calT_f(\bz_i) \setminus (P(y) \cup \bU_i)} T.\]

\item Let $\bj \geq i$ denote the first time at or after $i$ that $\bpi(\bj) \in R$.

\end{itemize}

Our argument considers two cases based on the size of the above set $\bigcup_{T \in  \calT_f(\bz_i) \setminus (P(y) \cup \bU_i)} T$.

\paragraph{\textsc{Case I:}} $|\bigcup_{T \in  \calT_f(\bz_i) \setminus (P(y) \cup \bU_i)} T| \geq k \sqrt{|\calT_f(\bz_i) \setminus (P(y) \cup \bU_i)|}$.
    Intuitively, in this case there are many ``safely flippable variables'' (recall the discussion after \Cref{eq:keyset}). Roughly speaking, in this case we will use the fact that the variable chosen at round $\bj$ is quite unlikely to be in the protected set (see \Cref{claim:missP}) to show that we output a valid stem either at round $\bj+1$ or at a later round with high enough probability to give the inductive statement.

    Let us show that $\bpi(\bj)$ is very likely to miss the protected set:
    \begin{claim} \label{claim:missP}
        \[ \Pr [\bpi(\bj) \not \in P(y)] \geq \left( 1 - \frac{1}{\sqrt{m}} \right). \]
    \end{claim}

    \begin{proof}
        Note that by \Cref{lem:stripped-indices-flippable}, we have $\bpi(t) \not \in R \setminus P(y)$ for all $t < i$. Thus
            \[\Pr \left [\bpi(\bj) \not \in P(y) \right] \geq \left(1 - \frac{|P(y)|}{|R|} \right). \]
        To finish the proof we then note that
        \[|R \setminus P(y)| \geq \left| \bigcup_{T \in  \calT_f(\bz_i) \setminus (P(y) \cup \bU_i)} T \right| \geq k\sqrt{m},\]
where the second inequality holds by the defining condition for Case~I, recalling that $|\calT_f(\bz_i) \setminus (P(y) \cup \bU_i)|=m$.  The first inequality holds because for each index there is a unique literal that could occur in a term appearing in the union. (Normally, there could be two such literals, but since $\bz_i$ satisfies all these terms there there is only one.)
    \end{proof}

    Now suppose that $\bpi(\bj) \not \in P(y)$. Since $\bpi(\bj) \in R$, it follows that $|\calT_f(\bz_{\bj+1}) \setminus (P(y) \cup \bU_{\bj+1})| < |\calT_f(\bz_{i}) \setminus (P(y) \cup \bU_{i})|$, where $\bU_{\bj+1}$ denotes the set of unanimous indices in $\calT_f(\bz_{\bj+1})$. Indeed, if we flipped $\bpi(\bj)$
    , then there was some stripped term  in $\calT_f(\bz_{i}) \setminus (P(y) \cup \bU_{i})$
    that contained a literal involving $\bpi(\bj)$, but by \Cref{cor:flipped-coords-missing}, no term in $\calT_f(\bz_{\bj+1})$ can contain a literal involving $\bpi(\bj)$. 
    If we did not flip $\bpi(\bj)$, then we must have that $\bpi(\bj) \in \bU_{\bj}$. However, $\bpi(j) \not \in \bU_i$, so we again conclude that some stripped term is no longer satisfied by $\bz_{j+1}$. So whether or not we flipped $\bpi(\bj)$,  \Cref{cor:stripped-terms-decrease} then gives us that $|\calT_f(\bz_{\bj+1}) \setminus (P(y) \cup \bU_{\bj+1})| < |\calT_f(\bz_{i}) \setminus (P(y) \cup \bU_{i})|$ as desired.

    Now  let
        \[p := \Pr \left[\text{we output a valid stem at time $\bj+1$} \bigg|  \bpi(\bj) \not \in P(y) \right]. \]
    Applying the inductive hypothesis at time $\bj + 1$, we get that
        \begin{align*}&\Pr \left[ \text{we output a valid stem after step $\bj+1$} \bigg| \bpi(\bj) \not \in P(y) \land \text{no valid stem at time $\bj + 1$} \right]\\
        &\geq e^{-100 \sqrt{m - 1} \log(k)}.
        \end{align*}

    Thus we have
    \begin{align*}
        \Pr \left[ \text{we output a valid stem after step $i$} \right] &\geq \Pr[\bpi(\bj) \not \in P(y)] \cdot \left(e^{-100 \sqrt{m - 1} \log(k)}(1-p) + p \right).  \\
    \intertext{As this is minimized when $p = 0$, we can lower bound this by}
        &\geq \Pr[\bpi(\bj) \not \in P(y)] \cdot \left(e^{-100 \sqrt{m - 1} \log(k)} \right) \\
        &\geq \left(1 - \frac{1}{\sqrt{m}} \right) \cdot \left(e^{-100 \sqrt{m - 1} \log(k)} \right) 
        \tag{by \Cref{claim:missP}}\\
        &\geq e^{-100 \sqrt{m - 1} \log(k)} \cdot e^{-2/\sqrt{m}} \\
        &\geq e^{-100 \left( \sqrt{m - 1} + \frac{0.02}{\sqrt{m}} \right) \log(k) } \\
        &\geq e^{-100 \sqrt{m} \log(k)}
    \end{align*}

    where we used that $1 - x \geq e^{-2x}$ for $0 \leq x \leq \frac{1}{\sqrt{2}}$ and $\sqrt{x} - \sqrt{x-1} \geq \frac{1}{2\sqrt{x}}$ for all $x \geq 1$. This completes the inductive step in Case~I.

\paragraph{\textsc{Case II:}} $|\bigcup_{T \in \calT_f(\bz_i) \setminus (P(y) \cup \bU_i)} T| \leq k \sqrt{|\calT_f(\bz_i) \setminus (P(y) \cup \bU_i)|}$. Intuitively, in this case the number of ``safely flippable variables'' is not very high.
Roughly speaking, we will use this to argue the existence of some ``popular index'' which, if it is chosen by $\bpi$ at round $\bj$, causes $|\calT_f(\bz_i) \setminus (P(y) \cup \bU_i)|$ to go down substantially (see \Cref{eq:godown}), and this will let us push the inductive claim through.

We begin by proving that such a ``popular index'' exists:
    \begin{claim}
    There exists a literal $\ell$ involving some index $a \in [n] \setminus (P(y) \cup \bU_i)$ that appears in at least $\sqrt{m}$ stripped terms in $\calT_f(\bz_i) \setminus (P(y) \cup \bU_i)$
    \end{claim}

    \begin{proof}
        We double count $(\ell, T)$ pairs, where $\ell$ is a literal and $T$ is a stripped term in $\calT_f(\bz_i) \setminus (P(y) \cup \bU_i)$ containing $\ell$. On the one hand, since we didn't output a valid stem, we must have that each stripped term has at least $k$ literals by \Cref{lem:genterm-output-stem}. Thus, there are at least 
            \[|\calT_f(\bz_i) \setminus (P(y) \cup \bU_i)| \cdot k\]
        many such pairs. On the other hand, let $c_{\max}$ denote the maximum number of stripped terms that contain a particular literal i.e. 
        \[c_{\max} := \max_{\ell} |\{T \in \calT_f(\bz_i) \setminus (P(y) \cup \bU_i): \ell \in T\}| \]
        
        The number of $(\ell,T)$ pairs is then bounded by
            \[c_{\max} \cdot k \sqrt{|\calT_f(\bz_i) \setminus (P(y) \cup \bU_i)|}\]
        Thus we conclude that
            \[c_{\max} \cdot k \sqrt{|\calT_f(\bz_i) \setminus (P(y) \cup \bU_i)|} \geq |\calT_f(\bz_i) \setminus (P(y) \cup \bU_i)| \cdot k\]
        Rearranging then yields the claim.
    \end{proof}
    
    Now, consider the literal $\ell$ promised by the claim and let $a$ be the index it involves. By \Cref{lem:stripped-indices-flippable}, $a$ has not yet occurred under $\bpi$ before time $i$. We then have that
    
    \[\Pr \left[ \bpi(\bj) = a \right] \geq \frac{1}{|R|} \geq \frac{1}{k + \left| \bigcup_{T \in \calT_f(\bz_i) \setminus (P(y) \cup \bU_i)} T \right|} \geq \frac{1}{2k^{3/2}}. \]

    We now note that $f(\bz_{\bj}^{\oplus a}) = 1$, implying that we will flip $a$ in round $\bj$. To see this, note that since (by definition of $\bj$) we only flip variables that are in $[n] \setminus R$ between the $i$th and $\bj$th iterations, $\calT_f(\bz_i) = \calT_f(\bz_j)$. (More formally, this can be shown by noting that flipping irrelevant variables does not change $\calT_f$, we don't flip variables in $P(y) \subseteq R$ by assumption, and variables in $\bU_i \setminus P(y)$ won't be flipped by \Cref{lem:uni-vars-flip-zero}.) Now by assumption, 
    $a \not \in \bU_{i}$; since  $\calT_f(\bz_i) = \calT_f(\bz_j)$ we have $\bU_i=\bU_j$, so we conclude $a \not \in \bU_j$. 
    Hence $\bpi(\bj)$ is flipped in the $\bj$th round if $\bpi(\bj) = a$. 

    Now invoking \Cref{cor:stripped-terms-decrease} and \Cref{cor:flipped-coords-missing} with the fact that $a$ appeared in at least $\sqrt{m}$ stripped terms, we get that 
\begin{equation} \label{eq:godown}
\text{if $\bpi(\bj) = a$ then
    $|\calT_f(\bz_{\bj+1}) \setminus (P(y) \cup \bU_{\bj+1})| < |\calT_f(\bz_{i}) \setminus (P(y) \cup \bU_{i})| - \sqrt{m}$}. 
\end{equation} 
    Applying the inductive hypothesis, we note that
    \begin{align*}
        \Pr \left [ \text{we output a valid stem after $\bj + 1$} \bigg| \bpi(\bj) = a \land \text{ no valid stem at $\bj +1$} \right] &\geq e^{-100 \sqrt{m - \sqrt{m}} \log(k)} \\     
        &\geq k^{50} e^{-100 \sqrt{m} \log(k)}
    \end{align*}

    where we used that $\sqrt{x} - \sqrt{x- \sqrt{x}} \geq 1/2$ for all $x \geq 1$. If we denote 
        \[p = \Pr \left[ \text{we output a valid stem at $\bj + 1$} \bigg | \bpi(\bj) = a \right] \]
    then it follows that we output a valid stem with probability at least
    \begin{align*}
        \frac{1}{2k^{3/2}} \left( p + (1 - p) k^{50} e^{-100 \sqrt{m} \log(k)} \right) &\geq \frac{k^{48.5}}{2} e^{-100 \sqrt{m} \log(k)} \geq e^{-100 \sqrt{m} \log(k)} \\
    \end{align*}
where the first inequality used that the minimum value occurs at $p = 0$ and the second inequality used that $k \geq 2$. This completes the inductive step in Case II, and hence the proof of \Cref{claim:inductive-find-stem}. \qedhere

\end{proof}

This concludes the proof of \Cref{lemma:find-stem-single-round}.
\end{proof}

\subsection{Proof of \Cref{lem:find-stem}}

With \Cref{lemma:find-stem-single-round} in hand it is easy to prove \Cref{lem:find-stem}:

\begin{proof}[Proof of \Cref{lem:find-stem}]
Consider some iteration of the while loop on line $2$ of \Cref{alg:find-stem}. 
Applying \Cref{lemma:find-stem-single-round}, it follows that we output a valid stem with probability at least $2^{-\wt{O}(\sqrt{k})}$. As we repeat the loop on line $2$ for $2^{\wt{O}(\sqrt{k})} \cdot \log(n)$ times, we conclude that we fail to output a valid stem with probability at most
        \[ \left(1 - 2^{-\wt{O}(\sqrt{k})} \right)^{\log(n) 2^{\wt{O}(\sqrt{k})}} = n^{-k}.\]
    Finally, note that each iteration of the while loop on line $2$ can add at most $n$ stems to the list $\calL_y$ of candidate stems. It follows that we output at most $n \log(n)  2^{\wt{O}(\sqrt{k})}$ candidate stems, as desired.
\end{proof}


\section{Expressing a $k$-term DNF as an ``augmented PTF''} \label{sec:structural}

The goal of this section is to establish \Cref{lem:augmented-PTF}, which shows that any $k$-term DNF formula can be expressed as an ``augmented PTF,'' of not-too-high degree and weight, over a suitable set of stems and auxiliary variables.  To explain this, we need some terminology.

\subsection{Setup and terminology}

\begin{definition} [Eligible pair] \label{def:eligible-pair}
An \emph{eligible pair} is a pair $(T',R_{T'})$ where $T'$ is a term (conjunction of Boolean literals over $x_1,\dots,x_n$) and $R_{T'} \subseteq [n]$ is a set of indices (auxiliary variables) which are \emph{disjoint} from $T'$, meaning that if $T'$ contains either $x_j$ or $\overline{x_j}$ then $j \notin R_{T'}.$
We will sometimes refer to $T'$ as a \emph{candidate stem}.
\end{definition}

\begin{definition} [Augmented monomial] \label{def:augmented-monomial}
Fix a set ${\cal F}=\{(T'_1,R_{T'_1}),\dots,(T'_N,R_{T'_N})\}$ of eligible pairs.  An \emph{${\cal F}$-augmented monomial of degree $d$} is a product of the form $T'_j \cdot \prod_{i \in S_j} x_i$ for some $j \in [N]$ and some $S_j \subseteq R_{T'_j}$, where $|S| \leq d$. 
\end{definition}

Recall that our variables $x_i$ take values in $\{0,1\}$, and that we view a term $T'$ as also being $\{0,1\}$-valued, so an ${\cal F}$-augmented monomial is a $\{0,1\}$-valued function.  We always assume that ${\cal F}$ contains an eligible pair $(T'_i,R_{T'_i})$ with $T'_i = $ the empty term (which corresponds to the constant-1 function).

\begin{definition} [Augmented PTF] \label{def:augmented-PTF}
Fix a set ${\cal F}$ as in \Cref{def:augmented-monomial}. An \emph{${\cal F}$-augmented PTF (polynomial threshold function) of degree $d$ and weight $W$} is a  Boolean function of the form 
\[
\mathbf{1}[p(x_1,\dots,x_n,T'_1,\dots,T'_N) \geq \theta],
\quad \text{where} \quad
p(x_1,\dots,x_n,T'_1,\dots,T'_N) = \sum_{j=1}^{N} w_{j} \cdot \left( T'_j \cdot  \prod_{i \in S_j} x_i\right)
\]
is a linear combination of ${\cal F}$-augmented monomials of degree $d$ (so each $S_j$ is a subset of $R_{T'_j}$) and each \emph{weight} $w_j$ is an integer, such that the total magnitude of all weights $\sum_{j=1}^N |w_j|$ is at most $W$.
\end{definition}

\begin{definition} [${\cal F}$ is fully expressive]
    \label{def:fully-expressive}
    Let $f$ be a $k$-term DNF over $x_1,\dots,x_n$.
A set ${\cal F}$ of (candidate stem $T',$ associated subset of variables $R_{T'}$) pairs is said to be \emph{fully expressive for $f$} if 
for each term $T_i$ of $f$,
 ${\cal F}$ contains a pair $(T',R_{T'})$ such that (a) $T'$ is a stem of $T_i$ and (b) $T_i \setminus T' \subseteq R_{T'}$.
\end{definition}

\subsection{Expressing a $k$-term DNF as an augmented PTF}

The following simple lemma plays an important role in our results:  it says that any $k$-term DNF can be expressed as a ``low-degree, low-weight'' augmented PTF.  Intuitively, this representation is what the Winnow2 algorithm will try to learn in \Cref{alg:Glorious-Learn-DNF}.

\begin{lemma} [Expressing a $k$-term DNF as an augmented PTF] \label{lem:augmented-PTF}
Let $f$ be a $k$-term DNF, and let ${\cal F}$ be fully expressive for $f$.  Then $f$ can be written as an augmented PTF over ${\cal F}$, of degree at most $\dmax$ and weight at most $\Wmax$, where
\begin{equation} \label{eq:dmax-and-Wmax}
\dmax := O(\sqrt{k \log k}),
\quad\quad\quad\quad
\Wmax :=  2^{O(\sqrt{k} \log^2(k))}.
\end{equation}
\end{lemma}

\begin{proof}
We closely follow the way that \cite{KlivansServedio:04jcss} uses Chebyshev polynomials to construct low-degree polynomial threshold functions for DNF formulas that are ``not too long.''  In our context the terms of our DNF $f$ may be arbitrarily long; we use the stems, together with the fact that ${\cal F}$ is fully expressive for $f$, to effectively reduce the length of the terms of $f$ to at most $2k$.

In more detail, write $f$ as $T_1 \vee \cdots \vee T_s$. Without loss of generality, for each $i \in [k]$ we may assume that the pair $(T'_i,R_{T'_i})$ is such that (a) $T'_i$ is a stem of $T_i$ and (b) $T_i \setminus T' \subseteq R_{T'}$.  This means that there is a subset $A_i \subseteq R_{T'_i}$ such that term $T_i$ can be written as $T'_i \wedge B_i$, where $B_i$ is a conjunction of length at most $2k$ that consists entirely of literals over the variables in $A_i$.

We will prove the following standard claim later:

\begin{claim} \label{claim:approx-conjunction}
There is an integer $0 < D \leq 2^{O(\sqrt{k} \log^2(k))}$ such that the following holds: For any conjunction $B$ of length at most $2k$, there is a polynomial $q_B$ of degree at most $\dmax$ such that 
\begin{equation} \label{eq:qbproperty}
q_B(x)  
\begin{cases}
= 1 & \text{~if~}B(x)=1;\\
\in \left[-{\frac 1 {2k}}, {\frac 1 {2k}} \right] & \text{~if~}B(x)=0.
\end{cases}
\end{equation}
for every $x \in \{0,1\}^n$. 
Moreover, $q_B$ is of the form $q'_B/D$, where $q'_B$ has integer coefficients whose absolute values sum to at most $ 2^{O(\sqrt{k} \log^2(k))}.$
\end{claim}

We now complete the proof of \Cref{lem:augmented-PTF} using \Cref{claim:approx-conjunction}.
By \Cref{claim:approx-conjunction}, 
\[
\text{the polynomial} \quad  \sum_{i=1}^k T'_i \cdot q_{B_i}(x)
\quad
\text{takes values}
\quad
\begin{cases}
\geq 1& \text{~if~}f(x)=1;\\
\in \left[-{\frac 1 2}, {\frac 1 2} \right] & \text{~if~}f(x)=0.
\end{cases}
\]
Since each $q_{B_i}$ has the same denominator of $D$, clearing the denominator we get that
\[
\text{the polynomial} \quad  q(x) := D \sum_{i=1}^k T'_i \cdot q_{B_i}(x)
\quad
\text{takes values}
\quad
\begin{cases}
\geq D& \text{~if~}f(x)=1;\\
\leq {\frac D 2} & \text{~if~}f(x)=0.
\end{cases}
\]
So the desired augmented polynomial threshold function computing $f$ is 
$\mathbf{1}\left[q(x) \geq {\frac {3D} 4} \right].$
It is clear that the degree of $q$ is at most $\dmax$, and it is easy to verify that the sum of the absolute values of the (integer) coefficients of $q$ is at most $k \cdot  \left(2^{O(\sqrt{k} \log^2(k))}\right)^2 =  2^{O(\sqrt{k} \log^2(k))}.$ This proves \Cref{lem:augmented-PTF} (up to the proof of  \Cref{claim:approx-conjunction}).
\end{proof}

\begin{proof} [Proof of \Cref{claim:approx-conjunction}] 
Fix $B$ to be any conjunction of length $m \leq 2k$.
Given a boolean literal $\ell$, its arithmetization is $x_i$ if $\ell = x_i$ and is $1-x_i$ if $\ell = \overline{x_i}$.  Let 
\[
S(x) = 2k-m + (\text{the sum of the arithmetizations of the $m$ literals in $B$),}
\] 
so $S(x)=2k$ if $B(x)=1$ and $0 \leq S(x)\leq 2k-1$ if $B(x)=0$.
We will show that the desired polynomial $q_B(x)$ is
\begin{equation} \label{eq:qB}
q_B(x) = {\frac 
{C_{\lceil \sqrt2k \rceil} \left(
{\frac {S(x)}2k} \cdot \left(1+{\frac 1 2k}\right)
\right)^{\log(2k)}}
{C_{\lceil \sqrt2k \rceil} \left(
1+{\frac 1 2k} \right)^{\log(2k)}}
},
\end{equation}
where $C_d$ is the degree-$d$ Chebyshev polynomial of the first kind.
The fact that $q_B$ satisfies \Cref{eq:qbproperty} is an easy  consequence of the following well-known property of Chebyshev polynomials \cite{Cheney:66}:

\begin{fact} \label{fact:Chebyshev}
The polynomial $C_d$ satisfies
\begin{itemize}
\item [(a)] $|C_d(t)| \leq 1$ for $|t| \leq 1$, with $C_d(1)=1$;
\item [(b)] $C'_d(t) \geq d^2$ for $t>1$, with $C'_d(1)=t^2.$
\end{itemize}
\end{fact}
It is clear that the degree of $q_B$ is at most $\lceil \sqrt2k \rceil \cdot \log(2k) \leq \dmax$ and that $q_B(x)=1$ if $B(x)=1$.  If $B(x)=0$, then since $S(x) \leq 2k-1$, we have that ${\frac {S(x)}2k} \cdot \left(1 + {\frac 1 2k}\right) \in [0,{\frac {(2k)^2 - 1}{(2k)^2}}] \subset [0,1]$, so the numerator of $q_B(x)$ is at most 1 and the denominator (using the derivative bound given by Item (b), together with the fact that $C_{\lceil \sqrt{2k} \rceil}(1)=1$) is at least $2^{\log(2k)} > 2k$. This gives \Cref{eq:qbproperty}.

For the ``Moreover'' sentence of \Cref{claim:approx-conjunction}, we recall the following fact, which is an easy consequence of the defining recurrence relation for the Chebyshev polynomials of the first kind:

\begin{fact} [Coefficient bounds for Chebyshev polynomials]
\label{fact:coefficients}
The degree-$d$ Chebyshev polynomial of the first kind $C_d(x) = \sum_{i=0}^d a_{i,d} x^i$ has integer coefficients which satisfy $\sum_{i=0}^d |a_{i,d}| \leq 3^d.$
\end{fact}

\Cref{fact:coefficients} implies that the denominator $C_{\lceil \sqrt{2k} \rceil} \left(
1+{\frac 1 2k} \right)^{\log(2k)}$ of $q_B(x)$ is a rational number $\alpha/\beta$ where $|\alpha|,|\beta|$ are both $(2k)^{O(\sqrt{2k} \cdot \log k)} = k^{O(\sqrt{k} \cdot \log k)}$.
Hence we can write 
\[
q_b(x) = 
C_{\lceil \sqrt{2k} \rceil} \left(
{\frac {S(x)(2k+1)}{{(2k)}^2}}\right)
^{\log(2k)} \cdot {\frac \beta \alpha}.
\]
\Cref{fact:coefficients} also implies that $C_{\lceil \sqrt{2k} \rceil} \left(
{\frac {S(x)(2k+1)}{(2k)^2}}\right)^{\log(2k)}$ is expressible as $a(x)/\gamma$ where $\gamma$ is an integer with $|\gamma|\leq (2k)^{2 \log(2k) \cdot \lceil \sqrt2k \rceil} = k^{O(\sqrt{k} \cdot \log k)}$ and $a(x)$ has integer coefficients whose absolute values sum to $(2k)^{O(\sqrt{k} \cdot \log k)}$. This straightforwardly yields the desired conclusion, and the claim is proved.
\end{proof}

\subsection{Discussion} \label{subsec:featuresdiscussion}

In light of \Cref{lem:augmented-PTF}, we make the following definition, which  captures the collection of features that will be used by the attribute-efficient Winnow2 linear threshold learning algorithm:

\begin{definition} [$\Features({\cal F})$]
    \label{def:features}
    Let ${\cal F}$ be a collection of pairs $(T',R_{T'})$ where each $T'$ is a candidate stem (i.e.~a conjunction of literals over $x_1,\dots,x_n$) and each $R_{T'}$ is a subset of $[n]$ that is disjoint from the variables that occur in $T'$.
    The set $\Features({\cal F})$ contains all ${\cal F}$-augmented monomials of degree at most $\dmax$.
    \end{definition}

\begin{remark} \label{rem:LTF-PTF-equivalence}
Note that an ${\cal F}$-augmented PTF of degree $\dmax$ and weight $\Wmax$ is equivalent to a linear threshold function over the feature set $\Features({\cal F})$ in which the weight vector has sum of magnitudes of (integer) weights at most $\Wmax$.
\end{remark}

\begin{remark} 
[The role of the $R_{T'}$ sets] 
\label{rem:RT-sets}
As suggested by the previous remark, when we are running Winnow2 using ${\cal F}$, the set of features (denoted $\Features({\cal F})$) for the linear threshold function will consist of all possible ${\cal F}$-augmented monomials of degree $\dmax$.
Writing $N$ to denote the number of pairs in ${\cal F}=\{(T'_1,R_{T'_1}),\dots,(T'_N,R_{T'_N})\}$, we have 
\begin{equation} \label{eq:bounds-on-Features}
{r \choose \leq \dmax} \leq |\Features({\cal F})| 
=
\sum_{i=1}^N {|R_{T'_i}| \choose \leq \dmax}
\leq |{\cal F}| \cdot {r \choose \leq \dmax},
\end{equation}
where $r = \max_{(T',R_{T'}) \in {\cal F}} |R_{T'}|$.
Since the running time of Winnow2 (even to perform a single update) is at least linear in the dimension of its feature space, this means that the running time of our learning algorithm, \Cref{alg:Glorious-Learn-DNF}, will be at least ${r \choose \leq \dmax}.$

The main point of \Cref{sec:find-relevant-variables} is to ensure that we can construct a fully expressive set ${\cal F}$ in which each set $R_{T'}$ is always of size at most $\Rmax = O(k^2 \log k)$.  
If we had not imposed the requirement that each ${\cal F}$-augmented monomial $T'_j \cdot \prod_{i \in S_j} x_i$  must have $S_j \subseteq R_{T'_j}$ (equivalently, if we had taken each $R_{T'_j}$ to simply be the entire set $[n]$, which amounts to ``getting rid of the $R_{T'_j}$'s''), then the left-hand side of \Cref{eq:bounds-on-Features} would be ${n \choose \dmax} \approx n^{\tilde{\Theta}(\sqrt{k})}$, and as discussed in \Cref{sec:techniques} this would be a lower bound on the number of features $|\Features({\cal F})|$ and hence on the running time of our approach.  
The point of using the $R_{T'}$ sets (and of the work we do in \Cref{sec:find-relevant-variables} to find them and argue that they will never be too large) is that with the upper bound $\Rmax = O(k^2 \log k)$ in hand, the upper bound $|{\cal F}| \cdot {r \choose \leq \dmax}$ on $|\Features({\cal F})|$ is at most $|{\cal F}| \cdot {O(k^2 \log k) \choose \dmax} = |{\cal F}| \cdot 2^{\tilde{O}(\sqrt{k})}$. This means that we have (at least potentially) ``gotten $n$ out of the base,'' at least as long as $|{\cal F}|$ is not too large.  (In \Cref{sec:overall-glorious}, specifically \Cref{UB-on-mistakes-when-fully-expressive}, we will upper bound ${\cal F}$.)
\end{remark}


\section{Finding relevant variables} \label{sec:find-relevant-variables}

In this section we give an efficient algorithm, called \FindRelevantVariable. The key property of this algorithm is that if it is given as input a valid stem $T'$ for some term of $f$, along with a suitable pair of positive and negative inputs, then it will ``grow'' the associated set $R_{T'}$ of auxiliary variables for $T'$ in a useful way.  Formally, we will establish the following:

\begin{lemma}
\label{lemma:findessentialvariables}
Suppose $f: \zo^n \to \zo$ is a $k$-term DNF, $T'$ is a valid stem of some term $T_i$ of $f$, and $\kappa \in (0,1)$. Let $R_{T'}$ be the set of morally relevant variables associated with $T'$ (cf. \Cref{def:morals}), and let $y$ and $z$ be positive and negative (resp.)\ assignments to $f$ both of which satisfy the stem $T'$.  
Suppose that the hybrid input
$z' := z_{R_{T'}} \sqcup y_{\overline{R_{T'}}}$ 
satisfies some short term of $f_{T'}$. Then 
\begin{enumerate}
    \item \FindRelevantVariable$(f,T', R_{T'}, y, z)$~runs in time $n \cdot \poly(k,\log(1/\kappa))$;
    \item 
    With probability at least $1-\kappa$, \FindRelevantVariable$(f,T', R_{T'}, y, z)$~updates $R_{T'}$ to contain one new morally relevant variable $x_j$ that was not previously in $T'$ or $R_{T'}$, where some term $T_{i'}$ that $T'$ is a stem of contains $x_j$ or $\overline{x_j}$.
\end{enumerate}
\end{lemma}

\subsection{Useful Preliminaries}
\label{subsec:noise-prelims}

A useful operator in the analysis of Boolean functions is the \emph{Bonami--Beckner noise operator} (see Chapter~2 of~\cite{o2014analysis} for more information):

\begin{definition} \label{def:noise}
	Fix $\rho\in[0,1]$. For a given $x\in\zo^n$, we write $\by\sim N_\rho(x)$ to mean a draw of $\by\in\zo^n$ where each bit $\by_i$ is drawn as follows: 	
	\[N_{\rho}(x) := \begin{cases}
	 x_i & \text{with probability}~\rho\\ 
	 1-x_i & \text{with probability}~1-\rho
	 \end{cases}.
	\]
	Given a function $f\isazofunc$, we define the \emph{noise operator} $\T_\rho$ as
	\[
		\T_\rho f(x) := \Ex_{\by\sim N_\rho(x)}\sbra{f(\by)}\,.
	\]
\end{definition}

For the rest of this section, we will fix 
\begin{equation} \label{eq:rho-def}
	\rho := 1 - \frac{1}{10\tau}
\end{equation}
where recall from \Cref{sec:preliminaries} that we set $\tau := 1000k$. 
Furthermore, throughout this section, let $g\isazofunc$ be a $k$-term DNF; recall that as discussed in \Cref{sec:preliminaries} we can write 
\[
	g := g_{\leq \tau} \vee g_{> \tau}.
\]

Ultimately, the role of $g$ will be played by $f_{T'}$ where $T'$ is a candidate stem. Note that when $T'$ corresponds to a bona fide stem of an actual term, it will be the case that $(f_{T'})_{\leq \tau} \neq \emptyset$. 

\begin{definition}[Term lengths] \label{def:term-lengths}
	We say that a term $T$ is:
	\begin{itemize}
		\item \emph{short} if its length is at most $\tau$; 
		\item \emph{medium} if its length is between $\tau+1$ and $1000\tau\log k$; and 
		\item \emph{long} if its length is strictly greater than $1000 \tau\log k$. 
	\end{itemize}
\end{definition}

The following easy claim shows how noise affects assignments that satisfy short terms: 

\begin{claim} \label{claim:pos}
Suppose $y^1 \in \zo^n$ is an input such that $g_{\leq \tau}(y^1)=1$.  Then $\T_\rho g(y^1) \geq 0.9.$
\end{claim}
\begin{proof}
	Let $T$ be the short term in $g$ that is satisfied by $y^1$. 
	We then have 
	\[
		\Prx_{\bz\sim N_\rho(y^1)}\sbra{T(\bz) = 0} \leq \tau\cdot\frac{1}{10\tau} \leq 0.1\,,
	\]
	thanks to our choice of $\rho$ as well as a union bound over the at most $\tau$ literals of $T$. The result follows immediately. 
\end{proof}

We will also require the following claim about how noise affects assignments that do not satisfy any short or medium terms:

\begin{claim} \label{claim:L16} 
	Suppose that $y^0$ is such that $g_{\leq 1000\tau \log k}(y^0)=0$. 
	Then $\T_\rho g(y^0)\leq 0.1.$
\end{claim}

\begin{proof}
	Let $P$ denote the protected set of $y$ (cf.~\Cref{def:protected}). 
	We have 
	\[
		\Prx_{\bz\sim N_\rho(y^0)}\sbra{\bz_i = y_i^0~\text{for all}~i\in P} = \pbra{1 - \frac{1}{10\tau}}^k \geq 0.99\,,
	\]
	where we used our choice of $\tau$ as well as the inequality $(1-a)^b \geq 1-ab$ for $a,b \geq 0$. 
	Call this event $E$; in other words, $E$ is the event that applying noise does not flip any of the indices in the protected set. 
	Let $T$ be a term of length at least $1000 \tau \log(k)$. 
	Note that 
	\[
		\Prx_{\bz \sim N_\rho(y^0)}\sbra{T(\bz) = 1 \mid  E} \leq \left(1 - \frac{1}{10\tau} \right)^{1000 \tau \log(k)} \leq k^{-100}\,.
	\]
	A union bound over at most $k$ such possible terms implies that 
	\[
		\Prx_{\bz\sim N_\rho(y^0)}\sbra{\exists~\text{long term}~T~\text{s.t.}~T(\bz) = 1 \mid E} \leq k^{-99}\,.
	\]
	It then follows that 
	\[
		\T_\rho g(y^{0}) = \Prx_{\bz\sim N_\rho(y^{0})}\sbra{g(\bz) = 1} \leq 0.01 + k^{-99} \leq 0.02\,,
	\]
	completing the proof. 
\end{proof}

The following terminology will be helpful:

\begin{definition}[Morally relevant variables]
\label{def:morals}
    Given a $k$-term DNF $g$, we say that $i \in [n]$ is \emph{morally irrelevant for $g$} if no short or medium term in $g$ contains either $x_i$ or $\overline{x_i}$, and will say that $i$ is \emph{morally relevant for $g$} otherwise. 
\end{definition}

The following claim shows that the value of $\T_\rho g(y)$ is not sensitive to flipping bits that are morally irrelevant for $g$:

\begin{claim} \label{claim:L14}
    Let $g$ be a $k$-term DNF and $y \in \zo^n$. 
    Then for any set $S\sse[n]$ consisting entirely of indices that are morally irrelevant for $g$, we have 
    \[
        \left|
        \T_\rho g (y) - \T_\rho g(y^{\oplus S})
        \right| 
        \leq {\frac 1 {k^{50}}}\,,
    \]
    where $y^{\oplus S} \in \zo^n$ is obtained by flipping the bits of $y$ whose indices are in $k$. 
\end{claim}

\begin{proof}
Note that we can equivalently view 
\[
    \T_\rho g(y) = \Ex_{\bz\sim N_\rho(y)}\sbra{f(\bz)} = \Ex_{\bz\sim N_\rho(0^n)}\sbra{f(y^{\oplus \bz})}\,,
\]
where we identify $\bz \in \zo^n$ with a subset of $[n]$ in the natural fashion. 
We then have 
\begin{align*}
\left|\T_\rho g(y) - \T_\rho g(y^{\oplus S})\right|
&=
\left|
\Ex_{\bz \sim N_\rho(0^n)}[g(y^{\oplus \bz})]
-
\Ex_{\bz \sim N_\rho(0^n)}[g(y^{\oplus S \oplus \bz})]
\right|\\
&\leq
\Ex_{\bz \sim N_\rho(0^n)} \left[
\left|
g(y^{\oplus \bz})
-
g(y^{\oplus S \oplus \bz})
\right|
\right] \tag{Jensen's inequality}\\
&=
\Prx_{\bz \sim N_\rho(0^n)} \left[
g(y^{\oplus \bz})
\neq
g(y^{\oplus S \oplus \bz})
\right]\\
&=\Prx_{\bz \sim N_\rho(0^n)} \left[
g(y^{\oplus \bz})=1,
g(y^{\oplus S \oplus \bz})=0
\text{~or~}
g(y^{\oplus \bz})=0,
g(y^{\oplus S \oplus \bz})=1
\right]\\
&\leq
2k \cdot \max_{x \in \zo^n} \,
\max_{\text{long term}~T~\text{in}~g} \,
\Prx_{\bz \sim N_\rho(0^n)}
\sbra{T(x^{\oplus \bz})=1}\,,
\end{align*}
where the last inequality is by a union bound over the (at most $k$) long terms in $g$, using the fact that no variable in $k$ occurs in any short or medium term of $g$.

For concreteness, we may suppose (without loss of generality) that the long term achieving the max in the expression above is $x_1 \wedge x_2 \wedge \dots \wedge x_M$ where $M>1000\tau \log k.$  Now, if $x \in \zo^n$ has $x_i=0$ for $j$ different indices $i \in [M]$, then since $\bz$ must flip precisely those indices and no others in $[M]$ to satisfy $T$, we have 
\begin{align*}
    \Prx_{\bz \sim N_\rho(0^n)}\sbra{T(x^{\oplus \bz})=1}
    &= (10\tau)^{-j}\left(1-{\frac 1 {10\tau}}\right)^{M-j} \\
    & \leq (10\tau)^{-j}\left(1-{\frac 1 {10\tau}}\right)^{1000\tau \log(k)-j} \\
    &= \pbra{10\tau - 1}^{-j}\pbra{1 - \frac{1}{10\tau}}^{1000\tau\log (k)} \\
    &\leq \pbra{10\tau - 1}^{-j}k^{-100} \\ 
    &\leq k^{-100}\,,
\end{align*}
which completes the proof. 
\end{proof}

Our algorithm will not compute the exact value of $\T_\rho g(x)$; instead, it will use an estimate of this value.
Let $\kappa \in (0,1)$ be a parameter we will set later, and let $\wh{\T}_\rho g(x)$ be an empirical estimate of $\T_\rho g(x)$ obtained by drawing $O(k^{200}\log \kappa^{-1})$ samples $\bz \sim N_\rho(x)$ and averaging $f(\bz)$. A standard Hoeffding bound gives that 
\[
    \Prx\sbra{\abs{\wh{\T}_\rho g(x) - \T_\rho g(x)} \geq k^{-100}} \leq \kappa\,.
\]

\begin{assumption} \label{noise-no-failure-assumption}
    By incurring failure probability $\kappa\cdot(\#~\text{times we estimate}~\wh{\T}_\rho g)$, we will assume that $\wh{\T}_\rho g(x)$ is always within an additive $\pm k^{-100}$ of $\T_\rho g(x)$. (We will set $\kappa$ in \Cref{sec:handling-failure-probability}.)
\end{assumption}

\subsection{Finding Relevant Variables}

We now turn to prove~\Cref{lemma:findessentialvariables}. 

\begin{algorithm}[t]
\addtolength\linewidth{-2em}

\vspace{0.5em}

\textbf{Input:} MQ access to a $k$-term DNF $f: \zo^n \rightarrow \zo$, a candidate stem $T'$, a set $R_{T'}$ of auxiliary variables for the stem $T'$, an $n$-bit string $y$ such that $f(y)=1,T'(y)=1$, and an $n$-bit string $z$ such that $f(z)=0,T'(z)=1$. \\[0.25em]
\textbf{Output:} A set $R_{T'}$
\vspace{0.5em}

\FindRelevantVariable($f, T', R_{T'}, y,z$):

\vspace{0.5em}

\begin{enumerate}

    \item Let $z' := z_{R_{T'}} \sqcup y_{\overline{R_{T'}}}$. 

    \item Do a line search from $z$ to $z'$ to find an input $x$ and a coordinate $i$ such that 
    \[
        \wh{\T}_\rho f_{T'}(x^{\oplus j}) \geq \wh{\T}_\rho f_{T'}(x) + k^{-3}\,.
    \]
    If no such $x$ and $i$ exist, then return FAIL.

    \item Add $i$ to $R_{T'}$ and return $R_{T'}$. 

\end{enumerate}
\caption{An algorithm to find relevant variables of a stem $T'$}
\label{alg:find-relevant-vars}
\end{algorithm}

\begin{proof}[Proof of~\Cref{lemma:findessentialvariables}]
The running time bound is obvious from inspection of \Cref{alg:find-relevant-vars} and our definition of $\wh{\T}_\rho$.

    We now turn to establish Item~2 of \Cref{lemma:findessentialvariables}. 
    Consider the line search in Step~2 of the algorithm; let $a_0=z,a_1,\dots,a_m=z'$ be the sequence of distinct points in $\zo^n$ that are visited in the course of this search. 
    In each step of this walk some coordinate $j$ is flipped; we say that a coordinate flip $j$ is a \emph{bad flip} if $j$ is morally irrelevant for $f_{T'}$, and we say that coordinate flip $j$ is a \emph{good flip} if $j$ is  morally relevant for $f_{T'}$ (cf.~\Cref{def:morals}).  We call a maximal sequence of consecutive bad flips a \emph{bad stretch}, and a maximal sequence of consecutive good flips a \emph{good stretch}.  

    The entire walk from $a_0$ to $a_m$ can be partitioned into alternating stretches
    \[
        \text{(bad stretch) (good stretch) (bad stretch) (good stretch)} \dots
    \]
    and we observe that the total number of flips occurring across all of the good stretches is at most $O(k^2 \log k)$ (because for a variable to be morally relevant for $f_{T'}$ it must be in a term of length at most $1000\tau \log k$ in $f_{T'}$, and there are at most $1000\tau k\log(k) = O(k^2 \log k)$ such variables). This in turn implies that there are at most $O(k^2 \log k)$ many good stretches, and at most that many bad stretches.

    Now, if $a_i$ is at the beginning of a bad stretch and $a_{i'}$ is at the end of a bad stretch, by \Cref{claim:L14} we know that $\T_\rho f_{T'}(a_i)$ and $\T_\rho f_{T'}(a_{i'})$ differ by at most $k^{-50}$.  So the sum of $\T_\rho f_{T'}(a_{k+1}) - \T_\rho f_{T'} (a_k)$, summed across all $k$ for which coordinate flip $k$ is a good flip, is at least $0.8 - k^{-50} \cdot O(k^2 \log k) > 0.7$; since there are at most $O(k^2 \log k)$ many good flips in total, there must be some good flip $k$ such that this difference is at least 
    \[
        \Omega\pbra{{\frac 1 {k^2 \log k}}} > \frac{1}{k^3}\,.
    \]
    Thanks to~\Cref{noise-no-failure-assumption}, we assume $\wh{\T}_\rho f_{T'}(\cdot) = \T_\rho f_{T'}(\cdot) \pm k^{-100}$, so the above inequality implies that  \FindRelevantVariable~will not fail.  
    Moreover, it follows from \Cref{claim:L14} that $i$ is morally relevant, which completes the proof.  
\end{proof}

\begin{remark} \label{remark:Rmax} It immediately follows from Item~2 of~\Cref{lemma:findessentialvariables} that for a valid stem $T'$ we have $|R_{T'}| \leq O(k^2 \log k)$, since every coordinate added to $R_{T'}$ is a morally relevant coordinate for $f_{T'}$ and as observed in the proof of \Cref{lemma:findessentialvariables}, there are at most $O(k^2 \log k)$ such coordinates. In the next section, we will write $R_{\max} = \Theta(k^2 \log k)$ to refer to this upper bound.
\end{remark}


\section{The overall learning algorithm} 
\label{sec:overall-glorious}

Throughout this section, up until \Cref{sec:handling-failure-probability} we make two assumptions saying that the subroutines we call throughout the algorithm will always give correct answers even though they each have a small probability of erring. 
The first is \Cref{noise-no-failure-assumption} above, and the second is
 \Cref{findcandidatestem-no-failure-assumption}:

\begin{assumption} \label{findcandidatestem-no-failure-assumption}
    By incurring a small failure probability, we will assume that \findcandidatestem$(f,y)$ \emph{always} outputs a valid stem for a term $T \in f$ that $y$ satisfies
when it is called on a $y \in \zo^n$ that satisfies $f_{>\tau}$ and not $f_{\leq \tau}$.
\end{assumption}

The rest of the second proceeds under these assumptions; we handle the straightforward failure probability analysis in \Cref{sec:handling-failure-probability} below by incorporating them into our bound on the overall failure probability of our algorithm.

We recall from \Cref{eq:dmax-and-Wmax} that the parameters $\dmax$ and $\Wmax$ are defined to be $\dmax := O(\sqrt{k \log k}),$
$\Wmax :=  2^{O(\sqrt{k} \log^2(k))}$.

\subsection{The main idea of the algorithm} 
\label{sec:main-idea}

The high level idea of our algorithm is to alternate between the following two phases of execution:

\begin{itemize}

\item [(1)] Running Winnow2 over a feature set $\Features({\cal F})$ of augmented monomials, which is based on a collection ${\cal F}$ of (candidate stem $T'$, associated subset of variables $R_{T'} \subset [n]$) pairs that the algorithm maintains (this feature set is defined below).

\item [(2)] Growing the feature set that is used to run Winnow.  This is done by either adding a new pair $(T',R_{T'})$ to ${\cal F}$, or growing the set $R_{T'}$ for some pair $(T',R_{T'})$ that is already present in ${\cal F}$.

\end{itemize}

As alluded to in (2) above, the set ${\cal F}$ grows in two ways throughout the execution of $\LearnDNF$. The first way is by adding new pairs (corresponding to adding new candidate stems); this is accomplished by the $\findcandidatestem$ procedure that was presented and analyzed in \Cref{sec:findstem}. The second way is by growing the set $R_{T'}$ for some stem $T'$ that is the first element of an existing pair that already belongs to ${\cal F}$; this is accomplished by the $\FindRelevantVariable$ procedure that was presented and analyzed in \Cref{sec:find-relevant-variables}. 
Intuitively, the algorithm's goal in growing ${\cal F}$ is for ${\cal F}$ to become \emph{fully expressive} as defined in \Cref{def:fully-expressive}.

Observe that by \Cref{lem:augmented-PTF}, if ${\cal F}$ is fully expressive then there is a low-weight, low-degree augmented PTF over ${\cal F}$, or equivalently, a linear threshold function over the features in $\Features({\cal F})$, that exactly computes $f$.  By \Cref{thm:Winnow2}, this implies that if Winnow2 is run over the feature set $\Features({\cal F})$ where ${\cal F}$ is fully expressive, it will succeed in exact learning the DNF $f$ after at most $\poly(\Wmax) \cdot \log |\Features({\cal F})|$ mistakes.

As our algorithm is running Winnow2, it repeatedly runs $\FindRelevantVariable$ in an attempt to find new relevant variables for the current stems. 
If $\FindRelevantVariable$ succeeds in finding a new relevant variable for any stem, then we abort the current run of Winnow2 and start over with the expanded version of $\Features({\cal F})$ (this is the second way that ${\cal F}$ can grow as mentioned above). 
On the other hand, if Winnow2~exceeds the number of mistakes that it would make if $\Features({\cal F})$ were fully expressive, then we run $\findcandidatestem$ to try to find a new candidate stem (this is the first way mentioned above that ${\cal F}$ can grow).

\subsection{The $\LearnDNF$ algorithm}

We now present our main DNF learning algorithm, \LearnDNF. The algorithm makes equivalence queries at various points (in line~4 and in the course of running Winnow2); if an equivalence query returns ``correct'' then the learning process halts (successfully).  

\bigskip

\begin{algorithm}[H]
\addtolength\linewidth{-2em}

\vspace{0.5em}

\textbf{Input:} MQ and EQ access to an unknown $k$-term DNF formula $f: \zo^n \rightarrow \zo$. \\[0.25em]
\textbf{Output:} A hypothesis $h: \zo^n \to \zo.$
\vspace{0.5em}

\LearnDNF($\MQ(f),\EQ(f)$):

\vspace{0.5em}

\begin{enumerate}

        \item Set 
        $
        M_{\max} := (\FCSM / n \log(n))^{\log^3 k} \cdot \log(n)
        $ (where $\FCSM$ is defined in \Cref{lem:find-stem}), and set $R_{\max} := \Theta(k^2 \log k)$ (as defined in \Cref{remark:Rmax}). 

        \item \label{initstep} Initialize ${\cal F}$ to contain the single pair $(T' := \emptyset,R_{T'}:=\emptyset)$. 
        Make an equivalence query on the constant-false function; assuming it returns a counterexample,  initialize the set $\ALLPOS$ to the set containing this counterexample.

        \item Initialize $POS = \emptyset$. Run Winnow2 over the feature set $\Features({\cal F})$, with the following checks each time a counterexample is received:

        \begin{enumerate}
            \item Each time Winnow2 receives a positive counterexample $z \in \zo^n$ to the hypothesis $h: \zo^n \to \zo$ that was just used for its equivalence query (so $h(z)=0$ but $f(z)=1$), add $z$ to both $POS$ and $\ALLPOS$.

            \item \label{lesscoollabel} Each time Winnow2 receives a negative counterexample to the hypothesis $h: \zo^n \to \zo$ that was just used for its equivalence query (so $h(z)=1$ but $f(z)=0$):
            \begin{itemize}
                \item For each pair $(T',R_{T'})$ in ${\cal F}$ with $T'(z)=1$ and $|R_{T'}|<R_{\max}$, find a point $y \in \ALLPOS$ with $T'(y)=1$  and run $\FindRelevantVariable(f, T', R_{T'}, y,z)$.
                \item If any of these calls to $\FindRelevantVariable$ updates any set $R_{T'}$ by adding a variable to it (and hence changes ${\cal F}$), then abort the current run of Winnow2 and return to the beginning of Step~3 (hence resetting $POS = \emptyset$ but leaving $\ALLPOS$ unchanged).
            \end{itemize}

            \item \label{Josh'sCoolLabel} (Magic Moment) If the current execution of Winnow2 has made $M_{\max}+1$ equivalence queries without having grown any set $R_{T'}$ for any pair $(T',R_{T'})$ in ${\cal F}$, then run $\findcandidatestem(f, y)$ for every $y \in POS$, add all of the resulting pairs that are returned by these runs to ${\cal F}$, then return to the beginning of Step~3 (hence resetting $POS = \emptyset$ but leaving $\ALLPOS$ unchanged).

        \end{enumerate}

\end{enumerate}
\caption{The Overall DNF Learning Algorithm}
\label{alg:Glorious-Learn-DNF}
\end{algorithm}

\bigskip \bigskip

\subsection{$M_{\max}$ is large enough to learn $k$-term DNFs}

As we now argue, $M_{\max}$ is an upper bound on the maximum number of mistakes that Winnow would ever make if it were run using a set $\Features({\cal F})$ where ${\cal F}$ was built over the course of the execution of \Cref{alg:Glorious-Learn-DNF} and ${\cal F}$ is fully expressive:

\begin{lemma} \label{lem:countfeatures}
Suppose that in the execution of LearnDNF on an $k$-term DNF $f$, step~\ref{Josh'sCoolLabel} (the Magic Moment) is reached at most $k+1$ times. Then, $|\Features({\cal F})| \leq n \log^2(n) 2^{\tilde{O}(\sqrt{k})} \cdot (\FCSM / n \log(n))^{2 + \log^3 k}$.
\end{lemma}

\begin{proof}
Let us first consider $|{\cal{F}}|$, i.e., the number of pairs in ${\cal{F}}$. At the initialization of the algorithm (step~\ref{initstep}), we have $|{\cal{F}}|=1$. Thereafter, pairs are only added to ${\cal{F}}$ in step~\ref{Josh'sCoolLabel} when we call $\findcandidatestem$.

By assumption, step~\ref{Josh'sCoolLabel} is reached at most $k+1$ times. The number of times we make a call to  
$\findcandidatestem$ in step~\ref{Josh'sCoolLabel} is $|POS|$, which is upperbounded by the number of mistakes that Winnow2 has made so far; by the condition which moves us to step~\ref{Josh'sCoolLabel}, this is at most $(k+1) \cdot (M_{\max}+1) \leq k^2 \cdot \log(n) \cdot (\FCSM / n \log(n))^{1 + \log^3 k}$ (with room to spare). Finally, each call to $\findcandidatestem$ may add up to $\FCSM = n \log(n) \cdot 2^{\tilde{O}(\sqrt{k})}$ pairs to ${\cal{F}}$ by \Cref{lem:find-stem}. We therefore have in total
$$|{\cal F}| \leq 1 + (k+1) \cdot |POS| \cdot \FCSM \leq k^4 n \log^2(n) \cdot (\FCSM / n \log(n))^{2 + \log^3 k}.$$

For each $(T', R_{T'}) \in \cal{F}$, we have $|R_{T'}| \leq R_{\max} = O( k^2 \log k)$ because of the size constraint imposed in the first bullet of step~\ref{lesscoollabel} of the algorithm. Thus, by the argument of \Cref{rem:RT-sets} above, we may conclude that 
$$|\Features({\cal F})| \leq |{\cal F}| \cdot 2^{\tilde{O}(\sqrt{k})} \leq  n \log^2(n) 2^{\tilde{O}(\sqrt{k})} \cdot (\FCSM / n \log(n))^{2 + \log^3 k},$$
as desired.
\end{proof}

The next lemma is where we use the attribute-efficiency of Winnow2, to show that our choice of $M_{\max}$ is sufficiently large.

\begin{lemma} \label{UB-on-mistakes-when-fully-expressive}
Suppose that in the execution of LearnDNF on an $k$-term DNF $f$, step~\ref{Josh'sCoolLabel} is reached at most $k+1$ times.  Let ${\cal F}$ be any fully expressive set that can be built over the course of such an execution of \Cref{alg:Glorious-Learn-DNF}.  Then if Winnow2 is run over the feature set $\Features({\cal F})$, it succeeds in exactly learning $f$ after at most $M_{\max}$ many equivalence queries (i.e.~mistakes).
\end{lemma}
\begin{proof}
By the correctness guarantee of Winnow2 from \Cref{thm:Winnow2}, and the definition of $\Wmax$ from \Cref{lem:augmented-PTF}, we know that Winnow2 will succeed as long as
$$M_{\max} \gg \poly(\Wmax) \cdot \log |\Features({\cal F})|.$$
\Cref{lem:augmented-PTF} shows that $\Wmax \leq 2^{O(\sqrt{k} \log^2(k))}$, and \Cref{lem:countfeatures} shows that $|\Features({\cal F})| \leq n \log^2(n) 2^{\tilde{O}(\sqrt{k})} \cdot (\FCSM / n \log(n))^{2 + \log^3 k}$. Therefore, using the bound $\FCSM = n \log(n) \cdot 2^{\tilde{O}(\sqrt{k})}$ from \Cref{lem:find-stem}, we can calculate
\begin{align*}\poly(\Wmax) \cdot \log |\Features({\cal F})| &\leq 2^{O(\sqrt{k} \log^2(k))} \cdot O(\sqrt{k} \cdot \polylog (k) + \log(n) + (\log^3 k) \cdot \log(\FCSM / n \log(n))) \\&\leq 2^{O(\sqrt{k} \log^2(k))} \cdot (\log n) \cdot \log(\FCSM / n \log(n)) \\&\leq 2^{O(\sqrt{k} \log^2(k))} \cdot (\log n)  \\& \ll \log(n) \cdot (\FCSM / n \log(n))^{\log^3 k} \\&= M_{\max},\end{align*}
as desired.\qedhere
\end{proof}

\subsection{Correctness of $\LearnDNF$}

Our goal in this subsection is to prove that our algorithm correctly learns $f$.

\begin{theorem} \label{lem:5b-at-most-s-times}
\Cref{alg:Glorious-Learn-DNF} correctly learns $f$, and step~\ref{Josh'sCoolLabel} is reached at most $k$ times during its execution.
\end{theorem}

Our proof of \Cref{lem:5b-at-most-s-times} will make use of the following three helper lemmas which describe the features and examples we must have when we reach different stages of the algorithm.

\begin{lemma} \label{helper1}
    Suppose step~\ref{Josh'sCoolLabel} has been reached at most $k+1$ times. 
    If ${\cal F}$ has a valid stem for each term of $f$, but for some term $T_i$, each valid stem $T'$ in ${\cal F}$ for $T_i$ is such that its $R_{T'}$ is missing a relevant variable, then one of the calls to $\FindRelevantVariable$ in step~\ref{lesscoollabel} will succeed and abort the current run of Winnow2 before step~\ref{Josh'sCoolLabel} is reached.
\end{lemma}

\begin{lemma} \label{helper2}
    Suppose step~\ref{Josh'sCoolLabel} has been reached at most $k+1$ times.     Each time we reach step \ref{Josh'sCoolLabel}, there must be a term of $f$ for which our ${\cal F}$ does not contain a valid stem.
\end{lemma}

\begin{lemma} \label{helper3}
    Suppose step~\ref{Josh'sCoolLabel} has been reached at most $k+1$ times.     If there is a term of $f$ for which ${\cal F}$ does not contain a valid stem, then before getting to step \ref{Josh'sCoolLabel}, we must receive a satisfying assignment $y$ such that
    \begin{itemize}
        \item $y$ satisfies a term $T_i$ for which ${\cal F}$ does not contain a valid stem, and
        \item $y$ satisfies $f_{>\tau}$ and not $f_{\leq \tau}$.
    \end{itemize}
\end{lemma}

We first show how to prove \Cref{lem:5b-at-most-s-times} using the helper lemmas, then we will prove the helper lemmas.

\begin{proof}[Proof of \Cref{lem:5b-at-most-s-times} assuming Lemmas~\ref{helper1}, \ref{helper2} and \ref{helper3}]

 Each of the first $k+1$ times that step \ref{Josh'sCoolLabel} is reached, \Cref{helper2} shows that there is a term of $f$ for which ${\cal F}$ does not have a valid stem.  \Cref{helper3} therefore shows that $POS$ must contain a point $y$ which satisfies a term for which ${\cal F}$ does not have a stem. \Cref{lem:find-stem} thus shows that $\findcandidatestem$ will add to ${\cal F}$ a stem for a new term, that it did not previously have a stem for, each of these times step \ref{Josh'sCoolLabel} is reached. In particular, since there are only $k$ terms, this means step \ref{Josh'sCoolLabel} was previously reached at most $k$ times. 

Towards a contradiction, let us suppose that \Cref{alg:Glorious-Learn-DNF} does not correctly learn $f$; under this assumption, we show that (after the last time step~\ref{Josh'sCoolLabel} is reached) ${\cal F}$ will eventually become fully expressive for $f$. First, ${\cal F}$ must contain a valid stem for each term in $f$ by \Cref{helper2,helper3}, since otherwise, we would again reach step~\ref{Josh'sCoolLabel}. Next, by \Cref{helper1}, we know that if there were a term of $f$ such that every valid stem in ${\cal F}$ is missing a relevant variable, then the run of Winnow2 would terminate in step~\ref{lesscoollabel}. Since we never run step \ref{Josh'sCoolLabel} again, it follows that the number of stems in $\calF$ is bounded. Since each stem can have at most $R_{\max} = O(k^2 \log(k))$ many associated relevant variables, it follows that we can only terminate Winnow2 in step \ref{lesscoollabel} a finite number of times. Once this upper bound has been reached, ${\cal F}$ must be fully expressive for $f$ as desired. Once $\calF$ becomes fully expressive for $f$, we can apply \Cref{UB-on-mistakes-when-fully-expressive}, which shows that when Winnow2 terminates, it will successfully return $f$ as desired. 
\qedhere

\end{proof}

\begin{proof}[Proof of \Cref{helper1}]

\emph Consider the  \emph{impostor} DNF $f'$, defined as follows. For each term $T_i$, let $E_i$ denote the set of its relevant variables, and consider a pair $(T'_i,R_{T'_i})$ from ${\cal F}$ such that $T'_i$ is a valid stem for $T_i$ and the set $R'_{T'_i} := R_{T'_i} \cap E_i$ is maximal. Define the term $\widetilde{T}_i$ to be the conjunction of all literals in $T_i$ corresponding to variables in $T'_i \cup R_{T'_i}$. In other words, $\widetilde{T}_i$ is gotten by removing from $T_i$ all literals corresponding to variables not in $T'_i \cup R_{T'_i}$.  Define
$$f'(x) := \bigvee_{i=1}^s \widetilde{T}_i.$$

The set ${\cal F}$ contains all the features necessary to learn $f'$ by its definition, so when we run Winnow2 for $M_{max}$ steps with the set of features from ${\cal F}$, we must encounter a $x$ for which $f(x) \neq f'(x)$. Indeed, if it never saw an input inconsistent with $f'(x)$, it would errantly output $f'(x)$ by \Cref{UB-on-mistakes-when-fully-expressive}.

In particular, this $x$ must satisfy $f(x) = 0$ and $f'(x) = 1$. This is because we defined $f'$ by removing literals from terms of $f$, so for any input point $z$, if $f'(z)=0$ then $f(x)=0$.

To conclude that that we would have exited in step~5a, it remains to argue that this negative example $x$ could play the role of $z$ in a $(T',R_{T'},y,z)$ quadruple that would lead $\FindRelevantVariable$ to update the set $R_{T'}$. We will do so by picking $T'$ and $y$, and then showing that together they satisfy the premises of \Cref{lemma:findessentialvariables}.

Since $f'(x)=1$, there is some particular term, say $\widetilde{T}_{i^*}$, for which $\widetilde{T}_{i^*}(x)=1$. Moreover, $T_{i^*}(x) = 0$ because $f(x) = 0$. The $T'$ we need is the valid stem of $T_{i^*}$ which we called $T'_{i^*}$ above.

We know that $\widetilde{T}_{i^*}(x)=1$ which implies $T'_{i^*}(x)=1$. Since $T'_{i^*}$ is a stem in our set ${\cal F}$, this means there must be a $y \in \ALLPOS$ such that $T'_{i^*}(y)=1$. This is because of the way we add stems to ${\cal F}$ in the algorithm: 
if $T'_{i^*} = \emptyset$ then any $y \in \ALLPOS$ will do (note that $\ALLPOS$ is nonempty because of the way it is initialized in step~\ref{initstep}) and otherwise, we must have added $T'_{i^*}$ to ${\cal F}$ from a call to $\findcandidatestem$ in step~\ref{Josh'sCoolLabel}, which itself takes as input a point $y$ that was in $POS$ at the time, and thus is still in $\ALLPOS$ now, and $\findcandidatestem$ only returns a stem that $y$ satisfies.

To apply \Cref{lemma:findessentialvariables}, it remains only to argue that the hybrid input $x_{R_{T'_{i^*}}} \sqcup y_{\overline{R_{T'_{i^*}}}}$ satisfies a short term of $f_{T'_{i^*}}.$
We claim that $T_{i^*}|_{T'_{i^*}}$ is a desired short term. This is a short term since $T'_{i^*}$ is a valid stem for $T_{i^*}$, so restricting to $T'_{i^*}$ leaves at most $2k$ remaining literals. The input $y$ satisfies $T_{i^*}$ by its very definition, and so it in particular satisfies the literals of $T_{i^*}$ that are contained in $\overline{R_{T'_{i^*}}}$. Finally, $x$ satisfies the literals of $T_{i^*}$ that are contained in $R_{T'_{i^*}}$ because $T'_{i^*}(x)=1$ as observed earlier.     
\end{proof}

\begin{proof}[Proof of \Cref{helper2}]

Each time we reach step \ref{Josh'sCoolLabel}, it means we went through the whole run of Winnow with its current $\Features({\cal F})$ set of features, making $M_{\max}+1$ many mistakes, without adding any more relevant variables to any $R_{T'}$.  

Assume to the contrary we have a valid stem for each term of $f$. If ${\cal F}$ were fully expressive, then by \Cref{UB-on-mistakes-when-fully-expressive} Winnow2 would have succeeded. Therefore, we are in the case that we have a valid stem for each term, but for some term $T_i$, each valid stem $T'$ for $T_i$ is such that its $R_{T'}$ is missing a relevant variable.  \Cref{helper1} thus says we would have exited in step~\ref{lesscoollabel}, which contradicts the premise that we reached step~\ref{Josh'sCoolLabel}.    
\end{proof}

\begin{proof}[Proof of \Cref{helper3}]

Assume to the contrary that we get to step~\ref{Josh'sCoolLabel} without ever getting such a satisfying assignment $y$. This means that every satisfying assignment $y$ we ever got satisfies some term $T_j$ for which ${\cal F}$ contains a stem (note that if $y$ satisfies $f_{\leq L}$, then $\emptyset$ is such a stem). Let $M \subset [k]$ be the set of indices of terms that we have a valid stem for, i.e., we have a valid stem for $T_i$ for all $i \in M$ and not for any $i \notin M$. Note that $M$ is a strict subset since, by assumption, there is a term of $f$ for which ${\cal F}$ does not have a stem.

Consider the following new \emph{impostor} DNF $f''$:
$$f'' := \bigvee_{i \in M} T_i.$$
This DNF $f''$ is different from $f$ since $M$ is a strict subset of $[k]$. However, we notice that every assignment we were given is consistent with $f''$. The assignments we got that satisfy $f$ also satisfy $f''$ since, as we assumed earlier, they all satisfy terms of $f$ for which we have a valid stem. The assignments we got that don't satisfy $f$ also don't satisfy $f''$  since $f''$ has a subset of the terms of $f$, meaning every satisfying assignment of $f''$ is also a satisfying assignment of $f$.

Since we didn't exit at step~\ref{lesscoollabel}, this means $f''$ is expressible by the ${\cal F}$ we currently have, so by \Cref{UB-on-mistakes-when-fully-expressive}, Winnow2 must give us a point $y$ with $f''(y) \neq f(y)$, a contradiction.\qedhere
    
\end{proof}

\subsection{Bounding the failure probability:  Getting rid of \Cref{noise-no-failure-assumption} and \Cref{findcandidatestem-no-failure-assumption}}  \label{sec:handling-failure-probability}

We finally prove our main result!

\begin{proof}[Proof of \Cref{thm:main}]
    We have thus far made use of \Cref{noise-no-failure-assumption} and \Cref{findcandidatestem-no-failure-assumption}. By \Cref{lem:5b-at-most-s-times}, under these assumptions, our \Cref{alg:Glorious-Learn-DNF} correctly learns $f$. As in the proof of \Cref{lem:countfeatures}, we know it has made $\leq \poly(n,k) \cdot 2^{\tilde{O}(\sqrt{k})}$ mistakes, and also at most this many calls to \findcandidatestem and $\FindRelevantVariable$. By the running time bounds of \Cref{lem:find-stem} and \Cref{lemma:findessentialvariables}, this means the total running time is also $\leq \poly(n,k) \cdot 2^{\tilde{O}(\sqrt{k})}$.

    Finally we may remove \Cref{noise-no-failure-assumption} and \Cref{findcandidatestem-no-failure-assumption} by union bounds: by the above argument, each of the subroutines is called at most $\poly(n,k) \cdot 2^{\tilde{O}(\sqrt{k})}$ times, and we may set the error probabilities to be much smaller than the inverse of this (by setting $\kappa = 1/(\poly(n,k) \cdot 2^{\tilde{\Theta}(\sqrt{k})})$ sufficiently small in \Cref{noise-no-failure-assumption}, and   by \Cref{lem:find-stem} for \Cref{findcandidatestem-no-failure-assumption}) so that overall, they will all succeed except with failure probability at most  $1/(\poly(n,k) \cdot 2^{\tilde{\Theta}(\sqrt{k})})$.
\end{proof}

\section*{Acknowledgements}
Josh Alman is supported in part by NSF Grant CCF-2238221 and a Packard Fellowship. 
Shyamal Patel is supported by NSF grants CCF-2106429, CCF-2107187, CCF-2218677, ONR grant ONR-13533312, and an NSF Graduate Student Fellowship. Rocco Servedio is supported by NSF grants CCF-2106429 and CCF-2211238. This work was partially completed while a subset of the authors was visiting the Simons Institute for the Theory of Computing.

\begin{flushleft}
\bibliographystyle{alpha}
\bibliography{allrefs}
\end{flushleft}

\end{document}